\pdfoutput=1
\documentclass{article}
\usepackage[final]{nips_2017}

\usepackage[usenames,dvipsnames,dvinames]{xcolor}
\usepackage{url,times,stfloats,flushend,float,crop,color,chngpage,array,alltt}
\usepackage{amsmath,amsthm,amsfonts,graphicx,amssymb,sidecap,subfigure,wrapfig,microtype,cancel,
  bbm}
\usepackage{algorithm}
\usepackage{algpseudocode}
\usepackage[font=footnotesize,labelfont=bf]{caption}

\newtheorem{theorem}{Theorem}

\newcommand{\dtv}{d}
\newcommand{\pepLength}{n}
\newcommand{\dripT}{T}

\newcommand{\pep}{x}
\newcommand{\candidatePeps}{D}
\newcommand{\pepDb}{\mathcal{D}}
\newcommand{\indicator}{\mathbf{1}}
\newcommand{\cA}{\mathcal{A}}
\newcommand{\cP}{\mathcal{P}}

\newcommand{\thomson}{\ensuremath{\mathsf{Th}}}
\newcommand{\obsSpec}{s}
\newcommand{\theo}{v}

\newcommand{\insPen}{a_{\mbox{mz}}}
\newcommand{\insPenInt}{a_{\mbox{in}}}
\newcommand{\meanInt}{\mu^{\mbox{in}}}
\newcommand{\mzTol}{w}

\newcommand{\Obsmz}{O^{\mbox{mz}}}
\newcommand{\Obsi}{O^{\mbox{in}}}
\newcommand{\mumz}{\mu^{\mbox{mz}}}
\newcommand{\insVar}{\bar{\sigma}^2}
\newcommand{\mzVar}{\sigma^2}
\DeclareMathOperator*{\argmax}{argmax}
\usepackage{times}
\title{Gradients of Generative Models for Improved Discriminative
  Analysis of Tandem Mass Spectra}
\author{{\bf John T. Halloran} \\
Department of Public Health Sciences \\
University of California, Davis \\
\texttt{jthalloran@ucdavis.edu} \\
\And
{\bf David M. Rocke} \\
Department of Public Health Sciences \\
University of California, Davis \\
\texttt{dmrocke@ucdavis.edu}
}

\hypersetup{bookmarks,colorlinks=true,citecolor=Violet,linkcolor=Mahogany,urlcolor=blue}
\begin{document}

\maketitle

\begin{abstract}
Tandem mass spectrometry (\emph{MS/MS}) is a high-throughput
technology used to identify the proteins in a complex biological
sample, such as a drop of blood.  A collection of spectra is generated
at the output of the process, each spectrum of which is representative
of a peptide (protein subsequence) present in the original complex
sample.  In this work, we leverage the log-likelihood gradients of
generative models to improve the identification of such spectra.  In
particular, we show that the gradient of a recently
proposed dynamic Bayesian network (DBN)~\cite{halloran2014uai-drip}
may be naturally employed by
a kernel-based discriminative classifier.  The resulting Fisher kernel
substantially improves upon recent attempts to combine
generative and discriminative models for post-processing analysis,
outperforming all other methods on the evaluated datasets.  We extend
the improved accuracy offered by the Fisher kernel framework to
other search algorithms by introducing Theseus, a DBN representing a
large number of widely used MS/MS scoring functions.  Furthermore,
with gradient ascent and max-product inference at hand, we use
Theseus to learn model parameters without any supervision.
\end{abstract}

\section{Introduction}
In the past two decades, tandem mass spectrometry (\emph{MS/MS}) has
become an indispensable tool for identifying the proteins present in a
complex biological sample.  At the output of a typical MS/MS
experiment, a collection of spectra is produced on the
order of tens-to-hundreds of thousands, each of which is
representative of a protein subsequence, called a \emph{peptide},
present in the original complex sample.  The main challenge in MS/MS
is accurately identifying the peptides responsible for generating each
output spectrum.  

The most accurate identification methods search a database of peptides
to first score peptides, then rank and return the top-ranking such
peptide.  The pair consisting of a scored candidate
peptide and observed spectrum is typically referred to as a
\emph{peptide-spectrum match} (PSM).  However, PSM scores returned
by such database-search methods are often difficult to compare across
different spectra (i.e., they are poorly calibrated), limiting the
number of spectra identified per search~\cite{keich2014importance}.
To combat such poor calibration, post-processors are typically used to
recalibrate PSM scores~\cite{kall:semi-supervised,
  spivak:improvements,  spivak:direct}.  

Recent work has attempted to exploit generative scoring functions for
use with discriminative classifiers to better recalibrate PSM scores;
by parsing a DBN's \emph{Viterbi path} (i.e., the most probable
sequence of random variables), heuristically derived features were
shown to improve discriminative recalibration using support vector
machines (SVMs).  Rather than relying on heuristics, we look towards
the more principled approach of a Fisher
kernel~\cite{jaakkolaFisherKernelNips1998}.  Fisher kernels
allow one to exploit the sequential-modeling strengths of generative
models such as DBNs, which offer vast design flexibility for
representing data sequences of varying length, for use with
discriminative classifiers such as SVMs, which offer superior accuracy
but often require feature vectors of fixed length.  Although the
number of variables in a DBN may vary given different observed
sequences, a Fisher kernel utilizes the fixed-length gradient of the
log-likelihood (i.e., the \emph{Fisher score}) in the feature-space of
a kernel-based classifier.  Deriving the Fisher scores of a DBN for
Rapid Identification of Peptides (DRIP)~\cite{halloran2014uai-drip},
we show that the DRIP Fisher kernel greatly improves upon the previous
heuristic approach; at a strict FDR of
$1\%$ for the presented datasets, the heuristically derived DRIP
features improve accuracy over the base feature set by an average
$6.1\%$, while the DRIP Fisher kernel raises this average improvement to
$11.7\%$ (Table~\ref{table:datasets} in
Appendix~\ref{appendix:impact}), thus
nearly doubling the total accuracy of DRIP post-processing.

Motivated by improvements offered by the DRIP Fisher kernel, we look
to extend this to other models by defining a generative model
representative of the large class of existing scoring
functions~\cite{craig:tandem, eng:approach, eng:comet, kim:msgfPlus,
  howbert:computing, wenger2013proteomics, mcilwain:crux}.  In
particular, we define a DBN (called \emph{Theseus}\footnote{In honor of
  Shannon's magnetic mouse, which could learn to traverse a small
  maze.})
which, given an observed spectrum,
evaluates the universe of all possible PSM scores.  In this work, we
use Theseus to model PSM score distributions with respect to the
widely used XCorr scoring function~\cite{eng:approach}.  The resulting
Fisher kernel once again improves discriminative
post-processing accuracy.  Furthermore, with the generative model in
place, we explore inferring parameters of the modeled scoring function
using max-product inference and gradient-based learning.  The
resulting coordinate ascent learning algorithm outperforms standard
maximum-likelihood learning.  Most importantly, this overall learning
algorithm is unsupervised which, to the authors' knowledge, is the
first MS/MS scoring function parameter estimation procedure not to
rely on any supervision.  We note that this overall training
procedure may be adapted by the many MS/MS search algorithms whose
scoring functions lie in the class modeled by Theseus.
The paper is organized as follows.  We discuss background information in
Section~\ref{section:background}, including the process by which MS/MS
spectra are produced, the means by which spectra are identified, and
related previous work.  In Section~\ref{section:dripFisherScores}, we
extensively discuss the log-likelihood of the DRIP model and derive
its Fisher scores.  In Section~\ref{section:theseus}, we introduce
Theseus and derive gradients of its log-likelihood.  We then discuss
gradient-based unsupervised learning of Theseus parameters and present
an efficient, monotonically convergent coordinate ascent algorithm.
Finally, in Section~\ref{section:results},
we show that DRIP and Theseus Fisher kernels substantially improve
spectrum identification accuracy and that Theseus' coordinate ascent
learning algorithm provides accurate unsupervised parameter
estimation.

\section{Background}
\label{section:background}
\begin{figure*}
\centering
\includegraphics[trim=0.9in 0.0in 1.0in 0.4in,clip=true,
width=0.48\textwidth]{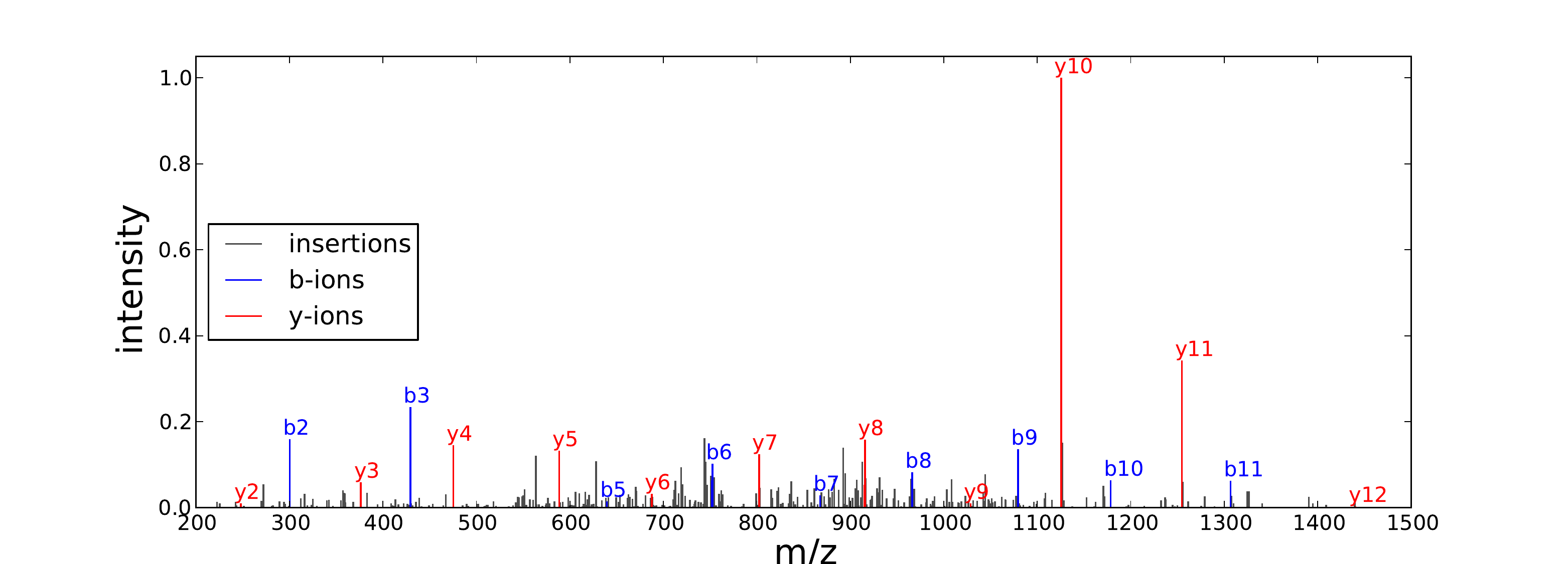}
\caption{{\small Example tandem mass spectrum with precursor charge
  $c(s)=2$ and generating peptide $\pep =
  \mbox{LWEPLLDVLVQTK}$.  Plotted in red and blue are, respectively,
  b- and y-ion peaks (discussed in
  Section~\ref{section:theoreticalSpectra}), while spurious observed
  peaks (called \emph{insertions}) are colored gray.  Note $y_1,
  b_1,b_4,$ and $b_{12}$ are absent fragment ions (called
  \emph{deletions}).}}
\label{fig:exampleSpectrum}
\end{figure*}
A typical tandem mass
spectrometry experiment begins by cleaving proteins into peptides
using a digesting enzyme.  The resulting peptides are then separated
via liquid chromatography and subjected to two
rounds of mass spectrometry.  The first round measures the mass and
charge of the intact peptide, called the \emph{precursor mass} and
\emph{precursor charge}, respectively.  Peptides are then fragmented
and the fragments undergo a second round of mass spectrometry,
the output of which is an observed spectrum indicative of the
fragmented peptide.  The x-axis of this observed spectrum denotes
\emph{mass-to-charge} (\emph{m/z}), measured in thomsons (\thomson),
and the y-axis is a unitless intensity measure, roughly proportional
to the abundance of a single fragment ion with a given m/z value.  A
sample such observed spectrum is illustrated in
Figure~\ref{fig:exampleSpectrum}.

\subsection{MS/MS Database Search}\label{section:databaseSearch}
Let $s$ be an observed spectrum with
precursor mass $m(s)$ and precursor charge $c(s)$.  In order to
identify $s$, we search a database of peptides, as follows.  
Let $\cP$ be the set
of all possible peptide sequences.  Each peptide
$x \in \cP$ is a string $x = x_1x_2 \dots x_{\pepLength}$ comprised of
characters, called \emph{amino acids}.  Given a
peptide database $\pepDb \subseteq \cP$, we wish to find the
peptide $\pep \in \pepDb$ responsible for generating $\obsSpec$.
Using the precursor mass and charge, the set of peptides to be scored
is constrained by setting a mass tolerance threshold, $\mzTol$, such
that we score the set of \emph{candidate peptides}
$\candidatePeps(s, \pepDb, \mzTol)= \left\{\pep: \pep
  \in \pepDb ,\, \left| \frac{m(\pep)}{c(s)}-m(s) \right| \leq
  \mzTol\right\}$, where $m(\pep)$ denotes the mass of peptide $x$.
Note that we've overloaded $m(\cdot)$ to return either a peptide's or
observed spectrum's precursor mass; we similarly overload $c(\cdot)$.
Given $s$ and denoting an arbitrary scoring function as
$\psi(\pep,s)$, the output of a search algorithm is thus $\pep^* =
\argmax_{\pep \in  \candidatePeps(m(s), c(s),\pepDb,
  \mzTol)} \psi(\pep,s)$, the top-scoring PSM.

\subsubsection{Theoretical Spectra}\label{section:theoreticalSpectra}
In order to score a candidate peptide $x$, fragment ions corresponding
to suffix masses (called \emph{b-ions}) and prefix masses (called
\emph{y-ions}) are collected into a \emph{theoretical spectrum}.  The
annotated b- and y-ions of the generating peptide for an observed
spectrum are illustrated in Figure~\ref{fig:exampleSpectrum}.  
Varying based on the
value of $c(s)$, the $k$th respective b- and y-ion
pair of $x$ are
\begin{equation*}
b(x,c_b, k) = \frac{\sum_{i = 1}^k m(x_i) + c_b}{c_b}, 
\;\;\;\; y(x,c_y, k) =
\frac{\sum_{i = n-k}^n m(x_i) + 18+ c_y}{c_y},
\end{equation*}
where $c_b$ is the charge of the b-ion and $c_y$ is the charge of the
y-ion.  For $c(s) = 1$, we have $c_b = c_y = 1$, since these are the only
possible, detectable fragment ions.  For higher observed charge states
$ c(s) \geq 2$, it is unlikely for a single fragment ion to consume
the entire charge, so that we have $c_b + c_y = c(s)$,  where $c_b,
c_y \in [1, c(s)-1]$. The b-ion offset corresponds to the mass of a $c_b$
charged hydrogen atom, while the y-ion offset corresponds to the mass of a water
molecule plus a $c_y$ charged hydrogen atom.  

Further fragment ions
may occur, each corresponding to the loss of a molecular group off a
b- or y-ion.  Called \emph{neutral losses}, these correspond to a loss
of either water, ammonia, or carbon monoxide.  These fragment ions are
commonly collected into a vector $v$, whose elements are weighted
based on their corresponding fragment ion.  For instance,
XCorr~\cite{eng:approach} assigns all b- and y-ions a weight of 50 and
all neutral losses a weight of 10.

\subsection{Previous Work}

Many scoring functions have been proposed for use in search
algorithms.  They range from simple dot-product scoring functions
(X!Tandem~\cite{craig:tandem}, Morpheus~\cite{wenger2013proteomics}),
to cross-correlation based scoring functions
(XCorr~\cite{eng:approach}), to exact $p$-values over linear scoring
functions calculated using dynamic programming
(MS-GF+~\cite{kim:msgfPlus} and XCorr
$p$-values~\cite{howbert:computing}).  The recently introduced
DRIP~\cite{halloran2014uai-drip} scores candidate peptides without
quantization of m/z measurements and allows learning the
expected locations of theoretical peaks given high quality, labeled
training data.  In order to avoid quantization of the m/z
axis, DRIP employs a dynamic alignment strategy wherein two types of
prevalent phenomena are explicitly modeled:
spurious observed peaks, called \emph{insertions}, and absent
theoretical peaks, called \emph{deletions} (examples of both are
displayed in Figure~\ref{fig:exampleSpectrum}).  DRIP then uses
max-product inference to calculate the most probable sequences of
insertions and deletions to score candidate peptides, and was shown to
achieve state-of-the-art performance on a variety of datasets.


In practice, scoring functions are often \emph{poorly calibrated}
(i.e., PSM scores from different spectra are difficult to compare to
one another), leading to potentially identified spectra left on the
table during statistical analysis.
In order to properly recalibrate such PSM scores, several
semi-supervised post-processing methods have been
proposed~\cite{kall:semi-supervised,
  spivak:improvements,spivak:direct}.  The most popular such method is
Percolator~\cite{kall:semi-supervised}, which, given the output target
and decoy PSMs (discussed in Section~\ref{section:results}) of a
scoring algorithm and features detailing
each PSM, utilizes an SVM to learn a discriminative classifier between
target PSMs and decoy PSMs.  PSM scores are then recalibrated using
the learned decision boundary.

Recent work has attempted to leverage the generative nature of the DRIP
model for discriminative use by Percolator~\cite{halloran2016dynamic}.
As earlier mentioned, the output of DRIP is the most probable sequence
of insertions and deletions, i.e., the Viterbi path.  However,
DRIP's observations are the sequences of observed spectrum
m/z and intensity values, so that the lengths of PSM's Viterbi paths
vary depending on the number of observed spectrum peaks.  In order to
exploit DRIP's output in the feature-space of a discriminative
classifier, PSM Viterbi paths were heuristically mapped to a
fixed-length vector of features.  The resulting heuristic features
were shown to dramatically improve Percolator's ability to
discriminate between PSMs.


\subsection{Fisher Kernels}
Using generative models to extract features for discriminative
classifiers has been used to great effect in many problem
domains by using Fisher
kernels~\cite{jaakkolaFisherKernelNips1998, jaakkola1999using,
  elkan2005deriving}.  Assuming a
generative model with a set of parameters $\theta$ and
likelihood $p(O | \theta) = \sum_{H}p(O, H | \theta)$, where $O$ is
a sequence of observations and $H$ is the set of hidden variables, the
\emph{Fisher score} is then $U_{o} =
\nabla_{\theta} \log p(O | \theta)$.  Given observations
$O_i$ and $O_j$ of differing length (and, thus, different
underlying graphs in the case of dynamic graphical models), a
kernel-based classifier over these instances is trained using
$U_{O_i}$ and $U_{O_j}$ in the feature-space.  Thus, a similarity
measure is learned in the gradient space, under the intuition that
objects which induce similar likelihoods will induce similar
gradients.
\section{DRIP Fisher Scores}\label{section:dripFisherScores}
\begin{figure}[htbp!]
\begin{center}
\includegraphics[page=8,trim=1.75in 1.2in 0.6in 0.0in,clip=true,
width=0.5\linewidth]{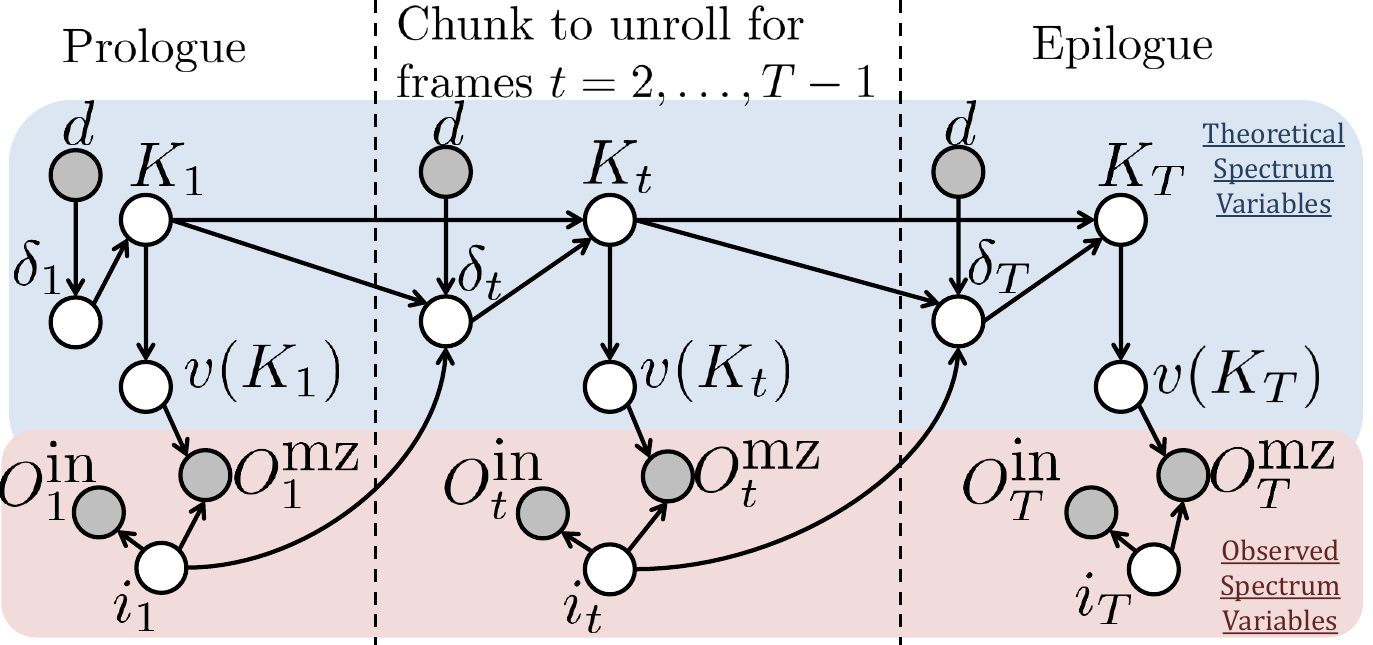}
\end{center}
\caption{{\small Graph of DRIP, the frames (i.e., time instances) of which correspond to
  observed spectrum peaks.  Shaded nodes represent observed variables
  and unshaded nodes represent hidden variables.  Given an observed
  spectrum, the middle frame (the chunk) dynamically expands to
  represent the second observed peak to the penultimate observed peak.}
}
\label{fig:dripGraph}
\end{figure}
We first define, in detail, DRIP's log-likelihood, followed by the
Fisher score derivation for DRIP's learned parameters.  For discussion
of the DRIP model outside the scope of this work, readers are directed
to~\cite{halloran2014uai-drip, halloran2016dynamic}.
Denoting an observed peak as a pair $(\Obsmz, \Obsi)$ consisting of an
m/z measurement and intensity measurement, respectively, let $s =
(\Obsmz_1, \Obsi_1), (\Obsmz_2, \Obsi_2), \dots, (\Obsmz_{\dripT},
\Obsi_{\dripT})$ be an MS/MS spectrum of $\dripT$ peaks and $x$ be a
candidate (which, given $s$, we'd like to score).  We denote the
theoretical spectrum of
$\pep$, consisting of its unique b- and y-ions sorted in ascending
order, as the length-$\dtv$ vector $\theo$.  The graph of DRIP is
displayed in Figure~\ref{fig:dripGraph}, where variables which control
the traversal of the theoretical spectrum are highlighted in blue and
variables which control the scoring of observed peak measurements are
highlighted in red.  Groups of variables are collected into time
instances called \emph{frames}.  The frames of DRIP correspond to the
observed peak m/z and intensity observations, so that there are $\dripT$
frames in the model.

Unless otherwise specified, let $t$ be an arbitrary frame $1 \leq t
\leq \dripT$.  $\delta_t$ is a multinomial random variable which dictates
the number of theoretical peaks traversed in a frame.  The random
variable $K_t$, which denotes the index of the current theoretical
peak index, is a deterministic function of its parents, such that
$p(K_t = K_{t-1} + \delta_{t} | K_{t-1}, \delta_t) = 1$.  Thus,
$\delta_t>1$ corresponds to the deletion of $\delta_t - 1$ theoretical
peaks.  The parents of $\delta_t$ ensure that DRIP does not attempt to
increment past the last theoretical peak, i.e., $p(\delta_t > \dtv -
K_{t-1}| \dtv, K_{t-1}, i_{t-1}) = 0$.  Subsequently, the theoretical peak value
$v(K_t)$ is used to access a Gaussian from a collection (the mean of
each Gaussian corresponds to a position along the m/z axis, learned
using the EM algorithm~\cite{dempster:maximum}) with which
to score observations.  Hence, the state-space of the model is all
possible traversals, from left to right, of the theoretical spectrum,
accounting for all possible deletions.

When scoring observed peak measurements, the Bernoulli random variable
$i_t$ denotes whether a peak is scored
using learned Gaussians (when $i_t = 0$) or considered an insertion
and scored using an insertion
penalty (when $i_t = 1$).  When scoring m/z observations, we thus have $p(\Obsmz_t | v(K_t),
i_t = 0) = f(\Obsmz_t | \mumz(v(K_t)), \mzVar)$ and $p(\Obsmz_t |
v(K_t), i_t = 1) = \insPen$, where $\mumz$ is a vector of Gaussian
means and $\mzVar$ the m/z Gaussian variance.  Similarly, when scoring intensity
observations, we have $p(\Obsi_t | i_t = 0) = f(\Obsi_t | \meanInt,
\insVar)$ and $p(\Obsi_t | i_t = 1) = \insPenInt$, where
$\meanInt$ and $\insVar$ are the intensity Gaussian mean and variance,
respectively.  Let $i_0 =  K_0 = \emptyset$ and $\indicator_{ \{ \cdot
\}}$ denote the indicator function.
Denoting DRIP's Gaussian parameters as $\theta$, the likelihood is
thus
{\small
\begin{align*}
p(s | x,\theta) &= \prod_{t = 1}^{\dripT} p(\delta_{t} | K_{t-1}, \dtv, i_{t-1})
\indicator_{\{ K_t = K_{t-1} + \delta_t \}}p(\Obsmz_t | K_t)p(\Obsi_t)\\
&= \prod_{t = 1}^{\dripT} p(\delta_{t} | K_{t-1}, \dtv, i_{t-1})
\indicator_{\{ K_t = K_{t-1} + \delta_t \}}(\sum_{i_t =
  0}^{1}p(i_t)p(\Obsmz_t | K_t, i_t))(\sum_{i_t =
  0}^{1}p(i_t)p(\Obsi_t | i_t))\\
&= \prod_{t = 1}^{\dripT} \phi (\delta_t, K_{t-1},
i_t, i_{t-1}).
\end{align*}}

The only stochastic variables in the model are $i_t$ and $\delta_t$, where
all other random variables are either observed or 
deterministic given the sequences $i_{1:\dripT}$ and $\delta_{1:\dripT}$.  Thus,
we may equivalently write $p(s | x,\theta) =
p(i_{1:\dripT}, \delta_{1:\dripT} | \theta)$.  The Fisher score of the
$k$th m/z mean is thus $\frac{\partial}{\partial \mumz(k)} \log p(s |
x,\theta) = \frac{1}{p(s | x,\theta)} \frac{\partial}{\partial
  \mumz(k)} p(s |x, \theta)$, and we have (please see
Appendix~\ref{appendix:dripFisherKernelDerivation} for the full derivation)
{\small
\begin{align}
\frac{\partial}{\partial
  \mumz(k)} p(s |x, \theta)=& 
\frac{\partial}{\partial \mumz(k)}
\sum_{i_{1:\dripT}, \delta_{1:\dripT}} p(i_{1:\dripT}, \delta_{1:\dripT} | \theta)
=
\sum_{i_{1:\dripT}, \delta_{1:\dripT} : K_t = k, 1 \leq
  t \leq \dripT} \frac{\partial}{\partial \mumz(k)}p(i_{1:\dripT}, \delta_{1:\dripT}
| \theta) \nonumber \\
=&
\sum_{i_{1:\dripT}, \delta_{1:\dripT}} \indicator_{\{ K_t = k \}}
p(s| x,\theta)
\left ( \prod_{t: K_t = k}\frac{1}{p(\Obsmz_t | K_t)} \right )
\left (\frac{\partial}{\partial \mumz(k)}  
\prod_{t: K_t = k}p(\Obsmz_t | K_t)
\right ). \nonumber
\end{align}
}
{\small
\begin{align}
\Rightarrow
\frac{\partial}{\partial \mumz(k)} \log p(s | x, \theta) 
=& \sum_{t = 1}^\dripT
p(i_{t}, K_{t} = k | s, \theta)
  p(i_t = 0 | K_t, \Obsmz_t)\frac{(\Obsmz_t -
    \mumz(k))}{\mzVar}\label{eqn:fisherScoreMzMean}.
\end{align}}
Note that the posterior in
Equation~\ref{eqn:fisherScoreMzMean}, and thus the Fisher score, may
be efficiently computed using sum-product inference.  Through similar
steps, we have
{\small
\begin{align}
\frac{\partial}{\partial \mzVar(k)} \log p(s | x,\theta) =&
 \sum_{t} 
p(i_{t}, K_{t} = k | s, \theta)
  p(i_t = 0 | K_t, \Obsmz_t)(\frac{(\Obsmz_t -
    \mumz(k))}{2 \sigma^4} - \frac{1}{2\mzVar}) 
\label{eqn:fisherMzVariance}\\
\frac{\partial}{\partial \meanInt} \log p(s | x,\theta) =&
\sum_{t} 
p(i_{t}, K_{t} | s, \theta)
  p(i_t = 0 | \Obsi_t)\frac{ (\Obsi_t -
    \meanInt)}{\insVar}
\label{eqn:fisherIntensityMean}\\
\frac{\partial}{\partial \insVar} \log p(s | x,\theta) =&
\sum_{t} 
p(i_{t}, K_{t} | s, \theta)
  p(i_t = 0 | \Obsi_t)
(\frac{(\Obsi_t -
    \meanInt)}{2 \bar{\sigma}^4} - \frac{1}{2\insVar}),
\label{eqn:fisherIntensityCovariance}
\end{align}
}
where $\mzVar(k)$ denotes the partial derivative of the variance for
the $k$th m/z Gaussian with mean $\mumz(k)$.

Let $U_{\mu} =
\nabla_{\mumz} \log p(s,x | \theta)$ and $U_{\sigma^2} =
\nabla_{\sigma^2(k)} \log p(s,x | \theta)$.  $U_{\mu}$ and $U_{\sigma^2}$ are length-$d$
vectors corresponding to the mapping of a peptide's sequence of b- and
y-ions into $r$-dimensional space (i.e., dimension equal to an
m/z-discretized observed spectrum).  Let $\mathbbm{1}$ be the
length-$r$ vector of ones.  Defining  $z^{\mbox{mz}}, z^{\mbox{i}} \in
\mathbb{R}^r$, the elements of which are the
quantized observed spectrum m/z and intensity values, respectively, we
use the following DRIP gradient-based features for SVM training in
Section~\ref{section:results}:
$|U_{\mu}|_1$, $|U_{\sigma^2}|_1$, $U_{\mu}^Tz^{\mbox{mz}}$,
$U_{\sigma^2}^Tz^{\mbox{i}}$, $U_{\mu}^T\mathbbm{1}$,
$U_{\sigma^2}^T\mathbbm{1}$, $\frac{\partial}{\partial \meanInt} \log
p(s,x | \theta)$, and $\frac{\partial}{\partial \insVar} \log p(s,x |
\theta)$.

\section{Theseus}
\label{section:theseus}
Given an observed spectrum $s$, we focus on representing the universe
of linear PSM scores using a DBN.  Let $z$ denote the
vector resulting from preprocessing the observed
spectrum, $s$.  As a modeling example, we look to represent the
popular XCorr scoring function.  Using subscript $\tau$ to denote a
vector whose elements are shifted $\tau$ units, XCorr's scoring
function is defined as
{\small
\begin{align*}
\mbox{XCorr}(s,x) &= v^T z - \sum_{\tau = -75}^{75} v^T z_{\tau}= v^T
(z - \sum_{\tau = -75}^{75} z_{\tau}) = v^T z',
\end{align*}}
where $z' = z - \sum_{\tau = -75}^{75} z_{\tau}$.  Let $\theta \in
\mathbb{R}^l$ be a vector of XCorr weights for the various types of
possible fragment ions (described in
Section~\ref{section:theoreticalSpectra}).  As described in
\cite{howbert:computing}, given $c(s)$, we
reparameterize $z'$ into a vector $z_{\theta}$ such that
$\mbox{XCorr}(x,s)$ is rendered as a dot-product between $z_{\theta}$
and a boolean vector $u$ in the reparameterized space.  This
reparameterization readily applies to any linear MS/MS scoring
function.  The $i$th element of $z_{\theta}$ is $z_{\theta}(i) = \sum_{j =
  1}^{l} \theta(j) z_j(i)$, where $z_j$ is a vector whose
element $z_j(i)$ is the sum of all higher charged fragment ions
added into the singly-charged fragment ions for the $j$th fragment ion
type.  The nonzero elements of $u$ correspond to the singly-charged
b-ions of $x$ 
and we have $u^Tz_{\theta} = \sum_{i =
  1}^nz_{\theta}(m(x_i) + 1) = \sum_{i = 1}^n \sum_{j =
  1}^l\theta(j) z_j(m(x_i) + 1) = v^Tz' = \mbox{XCorr}(s,x)$.

\begin{figure}[htbp!]
\begin{center}
\includegraphics[page=19,trim=2.3in 1.95in 2.55in 0.9in,clip=true,
width=0.5\textwidth]{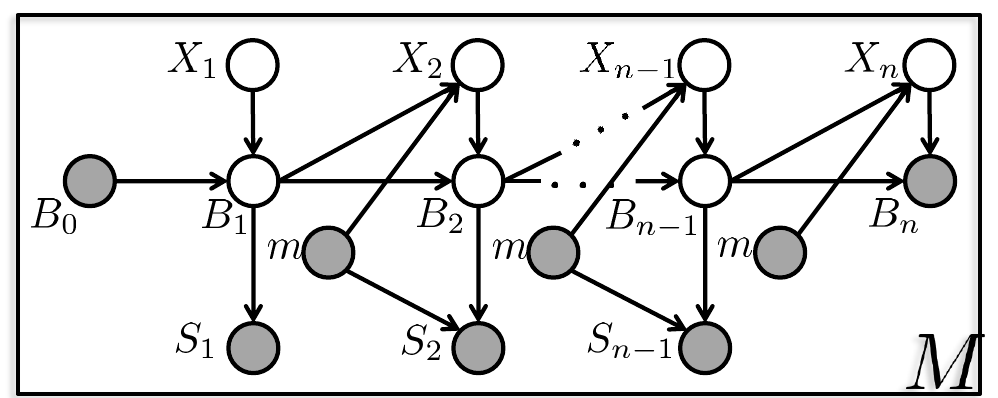}
\end{center}
\caption{{\small
Graph of Theseus.  Shaded nodes are observed
random variables and unshaded nodes are hidden (i.e., stochastic).
The model is unrolled for $n+1$ frames, including $B_0$ in frame zero.
Plate notation denotes $M$ repetitions of the model, where $M$ is the
number of discrete precursor masses allowed by the precursor-mass
tolerance threshold, $\mzTol$.}}
\label{fig:generativeModel}
\end{figure}

Our generative model is illustrated in
Figure~\ref{fig:generativeModel}.
$n$ is the maximum possible peptide length and $m$ is one of $M$
discrete precursor masses (dictated by the precursor-mass tolerance
threshold, $\mzTol$, and $m(s)$).  A
\emph{hypothesis} is an instantiation of random variables across all
frames in the model, i.e., for the set of all possible
sequences of $X_i$ random variables, $X_{1:n} = X_1, X_2, \dots, X_n$,
a hypothesis is $x_{1:n} \in X_{1:n}$.  In our case, each hypothesis
corresponds to a peptide and the corresponding log-likelihood its
XCorr score.  Each frame after the first contains an amino acid random
variable so that we accumulate b-ions in successive frames and access
the score contribution for each such ion.

For frame $i$, $X_i$ is a random amino acid and $B_i$ the
accumulated mass up to the current frame.  $B_0$ and $B_n$ are
observed to zero and $m$, respectively, enforcing the boundary
conditions that all length-$n$ PSMs considered begin with mass zero
and end at a particular precursor mass.  For $i > 0$, $B_i$ is a
deterministic function of its parents, $p(B_i | B_{i-1}, X_{i}) =
p(B_i = B_{i-1} +m(X_i)) = 1$.  Thus, hypotheses which do not respect
these mass constraints receive probability zero, i.e., $p(B_n \neq m |
B_{n-1}, X_n) = 0$.  $m$ is observed to the value of
the current precursor mass being considered.  

Let $\cA$ be the set of amino acids, where $|\cA| = 20$.  Given
$B_i$ and $m$, the conditional distribution of $X_i$ changes such that
$p(X_i \in \mathcal{A} | B_{i-1} < m) = \alpha \mathcal{U}\{
\mathcal{A} \}, p(X_i = \kappa | B_{i-1} \geq m) = 1$, where
$\mathcal{U}\{ \cdot \}$ is the uniform distribution over the input
set and $\kappa \notin \mathcal{A}$, $m(\kappa) = 0$.  Thus, when the
accumulated mass is less than $m$, $X_i$ is a random amino acid and,
otherwise, $X_i$ deterministically takes on a value with zero mass.
To recreate XCorr scores, $\alpha = 1 / |\mathcal{A}|$, though, in
general, any desired mass function may be used for $p(X_i
\in \mathcal{A} | B_{i-1} < m)$.


$S_i$ is a \emph{virtual evidence
  child}~\cite{pearl:probabilistic}, i.e., a leaf node
whose conditional distribution need not be normalized to compute
probabilistic quantities of interest in the DBN.  For our model, we
have  $p(S_i | B_i < m, \theta) = \exp(z_{\theta}(B_i)) = \exp(\sum_{i =
  1}^{|\theta|}\theta_i z_{i}(B_i))$ and $p(S_i | B_i \geq m, \theta)
= 1$.
Let $t'$ denote the first frame in which $m(X_{1:n}) \geq m$.  The
log-likelihood is then $\log p(s, X_{1:n} | \theta)$
{\small
\begin{align*}
&= \log p(X_{1:n},
B_{0:n}, S_{1:{n-1}})\\
& = \log( \indicator_{ \{B_0 = 0
\}} ( \prod_{i=1}^{n-1} p(X_i | m, B_{i-1}) p(B_i = B_{i-1} +
m(X_i)) p(S_i|  m, B_i, \theta)  ) \indicator_{ \{B_{n-1} + m(X_n)
= m\}})\\
& = \log \indicator_{ \{B_0 = 0 \; \wedge m(X_{1:n}) = m \}} +
\log (\prod_{i=t'+1}^{n} p(X_i | m, B_{i-1}) p(B_i = B_{i-1}
+ m(X_i)) p(S_i |  m, B_i, \theta)) + \\
& \;\;\;\; \log ( 
\prod_{i=1}^{t'} p(X_i | m, B_{i-1}) p(B_i = B_{i-1}
+ m(X_i)) p(S_i|  m, B_i, \theta)  )\\
& = \log \indicator_{ \{ m(X_{1:n})
= m \}} +
\log 1 + \log ( 
\prod_{i=1}^{t'} \exp(z_{\theta}(B_i))  )\\
& = \log \indicator_{ \{m(X_{1:n})
= m \}} + \sum_{i = 1}^{t'}z_{\theta}(B_i) = \log \indicator_{ \{B_0 = 0
\; \wedge \; m(X_{1:n}) = m \}} + \text{XCorr}(X_{1:t'}, s)
\end{align*}
}
The $i$th element of Theseus' Fisher score is thus
{\small
\begin{align}
\frac{\partial}{\partial \theta(i)} \log p(s | \theta) &=
\frac{\partial}{\partial \theta(i)} \log
\sum_{x_{1:n}}p(s, x_{1:n} | \theta)= \frac{1}{p(s
  | \theta)} \frac{\partial}{\partial \theta(i)}
\sum_{x_{1:n}}p(s, x_{1:n} | \theta) \nonumber\\
&= \frac{1}{p(s | \theta)} 
\sum_{x_{1:n}} \indicator_{ \{b_0 = 0 \; \wedge \;
  m(x_{1:n}) = m \} } (\sum_{j=1}^{t'}z_i(b_j))\prod_{j=1}^{t'}
\exp(z_{\theta}(b_j)) \label{equation:theseusFisher}
,
\end{align}
}
While Equation~\ref{equation:theseusFisher} is generally difficult to
compute, we calculate it efficiently using sum-product inference.
Note that when the peptide sequence is observed,
i.e., $X_{1:n} = \hat{x}$, we have $\frac{\partial}{\partial
  \theta(i)} \log p(s, \hat{x} | \theta) =
\sum_{j}z(m(\hat{x}_{1:j}))$.  

Using the model's Fisher scores,
Theseus' parameters $\theta$ may be learned via maximum likelihood
estimation.  Given a dataset of spectra $s^1, s^2, \dots, s^n$, we
present an alternate learning algorithm
(Algorithm~\ref{algorithm:coordinateAscent}) which
converges monotonically to a local optimum (proven in
Appendix~\ref{appendix:monotonicProof}).  Within each iteration,
Algorithm~\ref{algorithm:coordinateAscent} uses max-product inference
to efficiently infer the most probable PSMs per iteration,
mitigating the need for training labels.  $\theta$ is maximized in
each iteration using gradient ascent.

{\small
\begin{algorithm}
\caption{Theseus Unsupervised Learning Algorithm}\label{algorithm:coordinateAscent}
\begin{algorithmic}[1]
\While{not converged}
\For{$i = 1, \dots, n$}
\State $\hat{x}^i \leftarrow \argmax_{x^i \in \cP} \log p(s^i, x^i | \theta)$
\EndFor
\State $\theta \leftarrow \argmax_{\theta} \sum_{i=1}^n \log p(s^i, \hat{x}^i |
\theta)$
\EndWhile
\end{algorithmic}
\end{algorithm}
}

\section{Results}
\label{section:results}
Measuring peptide identification performance is complicated by the
fact that ground-truth is unavailable for real-world data.  Thus, in
practice, 
it is most common to estimate the \emph{false discovery rate}
(\emph{FDR})~\cite{benjamini:controlling} by searching a decoy
database of peptides which
are unlikely to occur in nature, typically generated by
shuffling entries in the
target database~\cite{keich2015improved}.  For a particular
score threshold, $t$, FDR is then calculated as the proportion of
decoys scoring better than $t$ to the number of targets scoring better
than $t$.  Once the target and decoy PSMs are calculated, a
curve displaying the FDR threshold vs. the number of correctly
identified targets at each given threshold may be calculated.  In
place of FDR along the x-axis, we use the
\emph{q-value}~\cite{keich2015improved},  defined to be the minimum
FDR threshold at which a
given score is deemed to be significant.  As many applications
require a search algorithm perform well at low thresholds, we only
plot $q \in [0, 0.1]$.  

The same datasets and search settings used to evaluate DRIP's
heuristically derived features in~\cite{halloran2016dynamic} are
adapted in this work.  MS-GF+ (one of the most accurate search
algorithms in wide use, plotted for reference) was run using version
9980, with PSMs ranked by E-value and Percolator features calculated
using \texttt{msgf2pin}.  All database
searches were run using a $\pm 3.0 \thomson$ mass tolerance, XCorr
flanking peaks not allowed in Crux searches, and all search
algorithm settings otherwise left to their defaults.  Peptides were
derived from the protein databases using trypsin cleavage rules
without suppression of proline and a single fixed carbamidomethyl
modification was included.

Gradient-based feature representations derived from DRIP and
XCorr were used to train an SVM classifier~\cite{kall:semi-supervised} and recalibrate PSM
scores.  Theseus training and computation of XCorr Fisher scores were
performed using a customized version of Crux
v2.1.17060~\cite{mcilwain:crux}.  For an XCorr PSM, a feature
representation is derived directly using both $\nabla_{\theta} \log
p(s | \theta)$ and $\nabla_{\theta} \log p(s,x | \theta)$ as defined
in Section~\ref{section:theseus}, representing gradient information
for both the distribution of PSM scores and the
individual PSM score, respectively.  
DRIP gradient-based
features, as defined in Section~\ref{section:dripFisherScores}, were
derived using a customized version of the DRIP
Toolkit~\cite{halloran2016dynamic}.
Figure~\ref{fig:absRanking}
displays the resulting search accuracy for
four worm and yeast datasets
.  For the uncalibrated search
results in Figure~\ref{fig:theseusUnsupervisedLearning}, we show that
XCorr parameters may be learned without supervision using Theseus, and
that the presented coordinate descent algorithm (which estimates the
most probable PSMs to take a step in the objective space) converges to
a much better local optimum than maximum likelihood estimation.
\begin{figure}[htbp!]
  \centering
  \subfigure{\raisebox{9.0mm}{\includegraphics[trim=3.4in 1.0in 0.3in 1.7in,
    clip=true,scale=0.485]{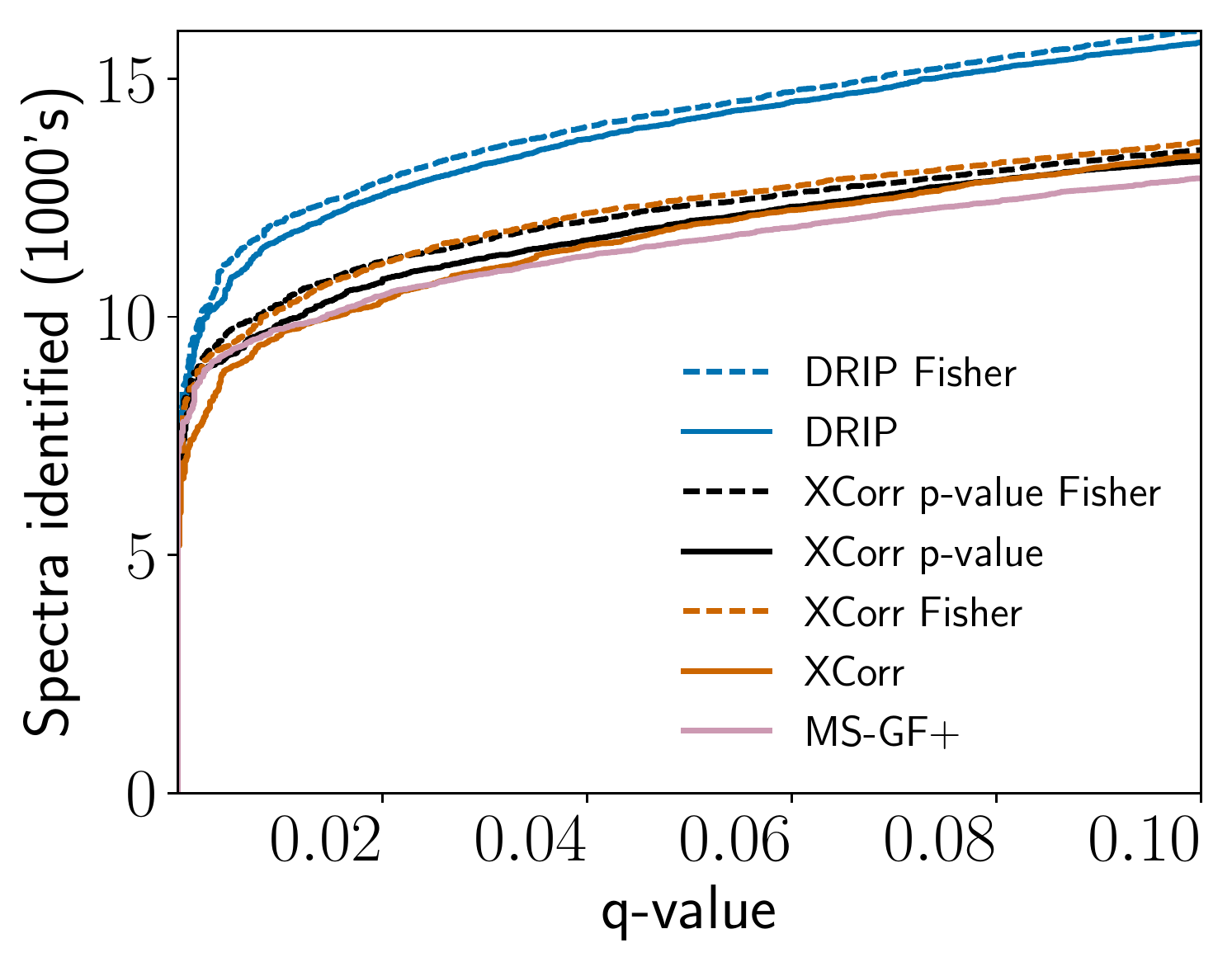}}}
  \subfigure[Worm-1]{\includegraphics[trim=0.0in 0.0in 0.0in 0.05in,
    clip=true,scale=0.3]{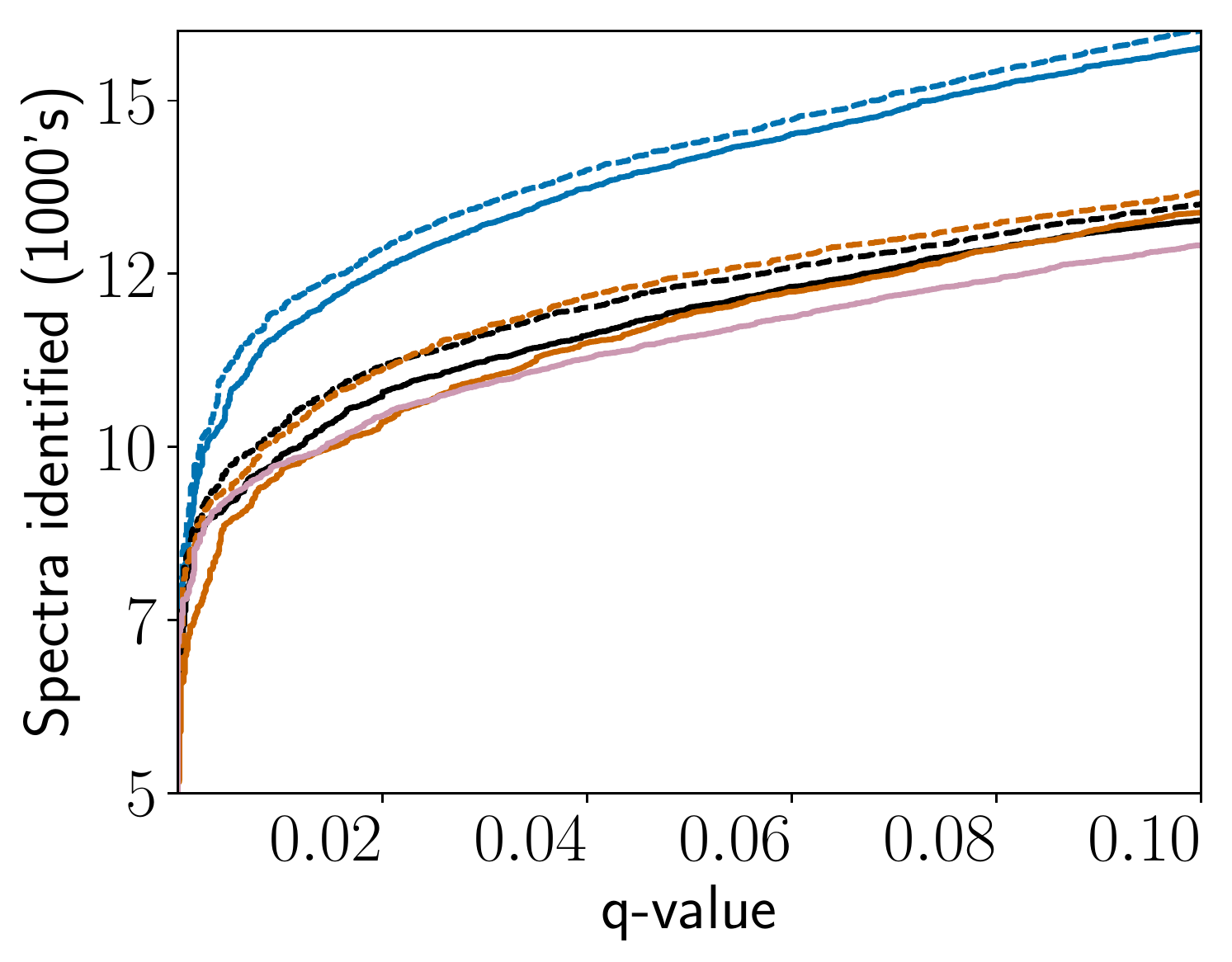}}
  \subfigure[Worm-2]{\includegraphics[trim=0.45in 0.0in 0.0in 0.05in,
    clip=true,scale=0.3]{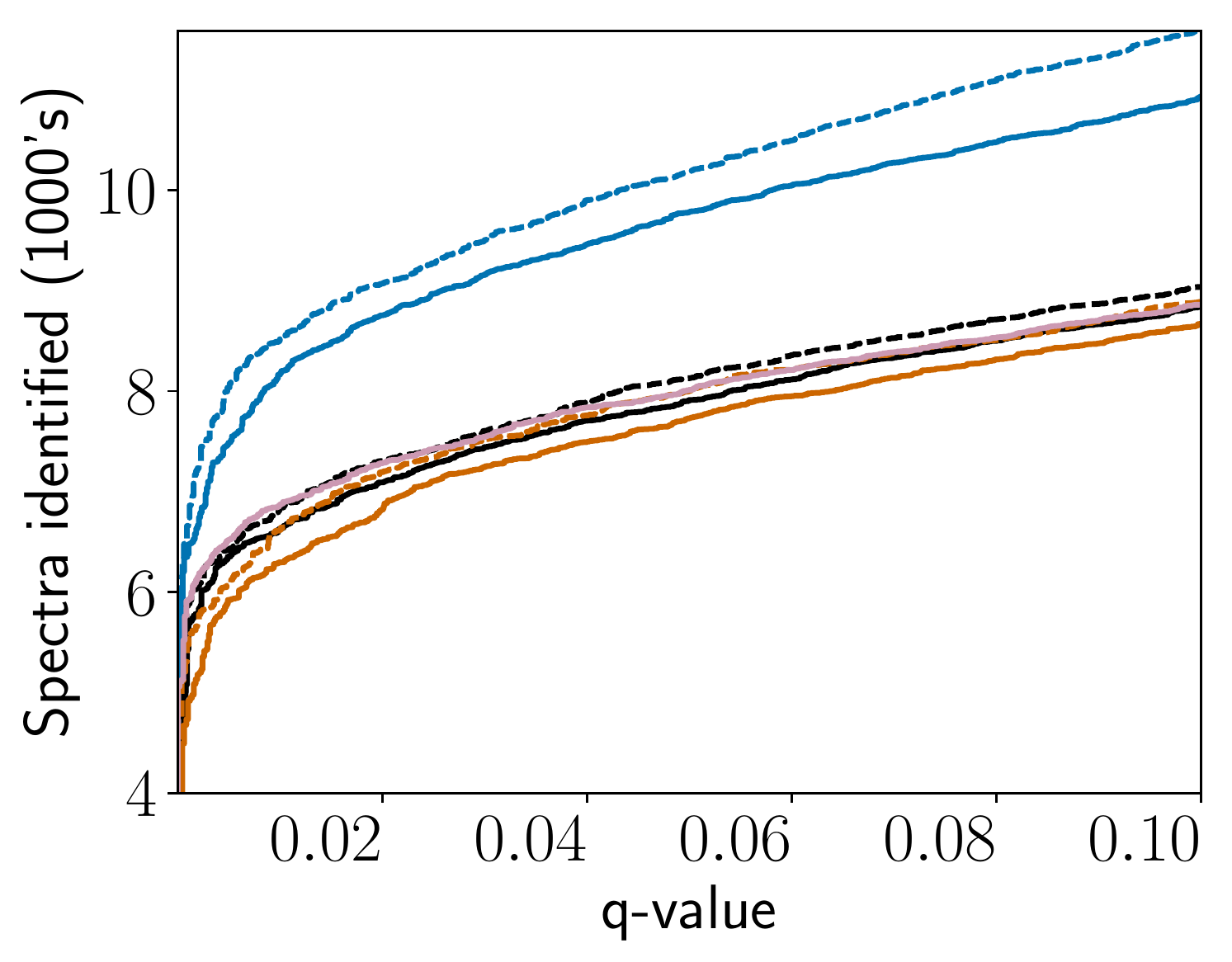}}
  \subfigure[Worm-3]{\includegraphics[trim=0.45in 0.0in 0.0in 0.05in,
    clip=true,scale=0.3]{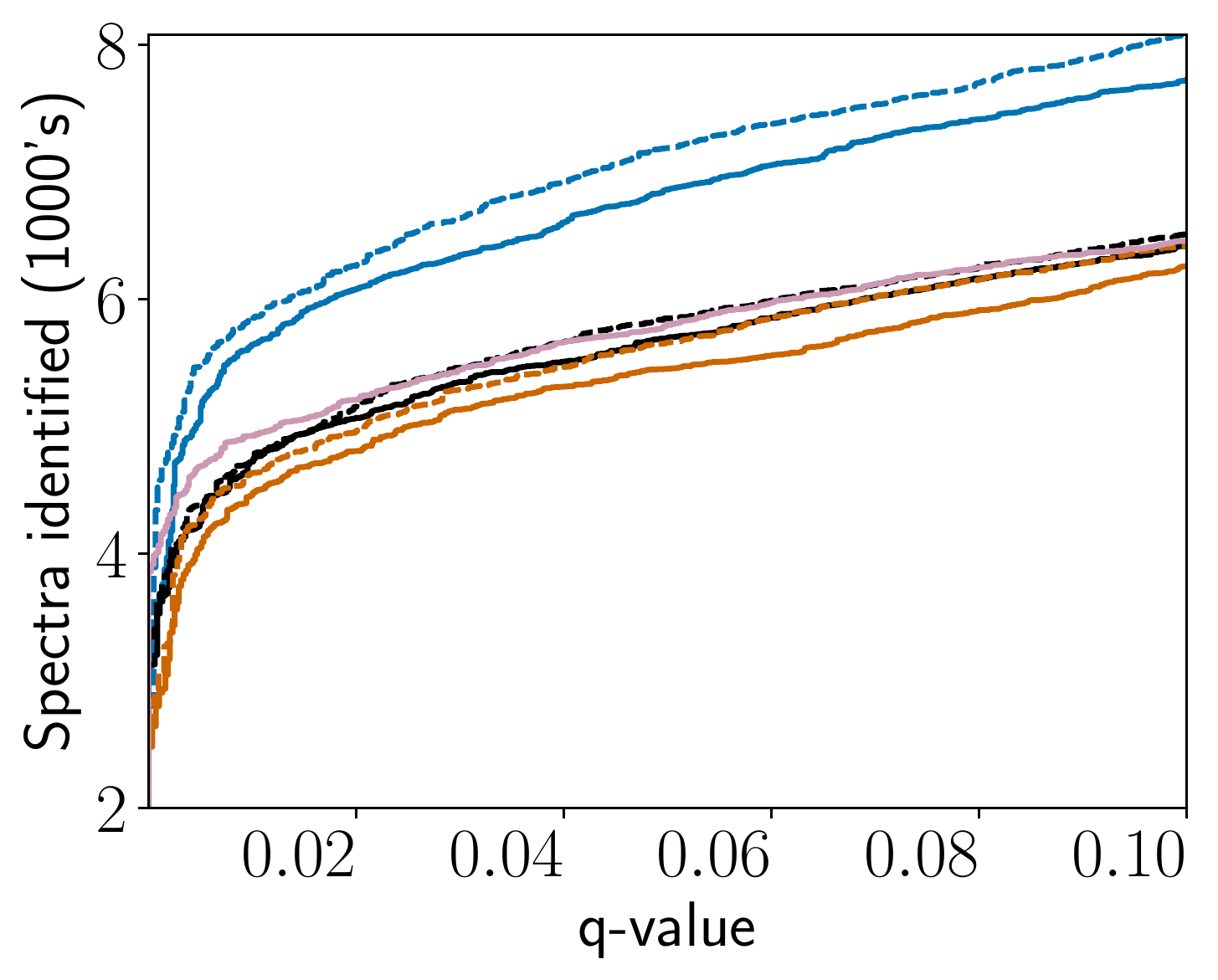}}
  \subfigure[Worm-4]{\includegraphics[trim=0.0in 0.0in 0.0in 0.05in,
    clip=true,scale=0.3]{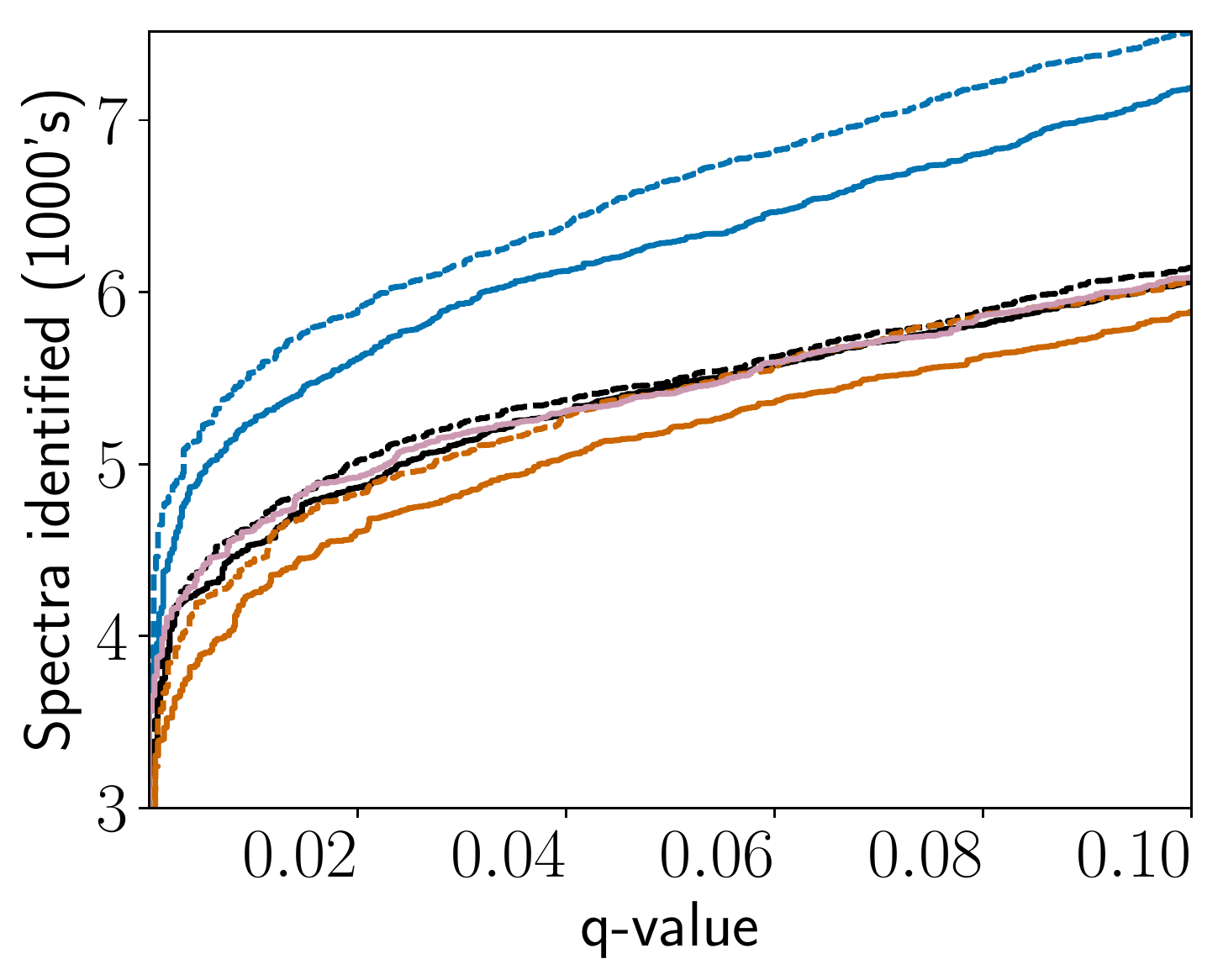}}
  \subfigure[Yeast-1]{\includegraphics[trim=0.45in 0.0in 0.0in 0.05in,
    clip=true,scale=0.3]{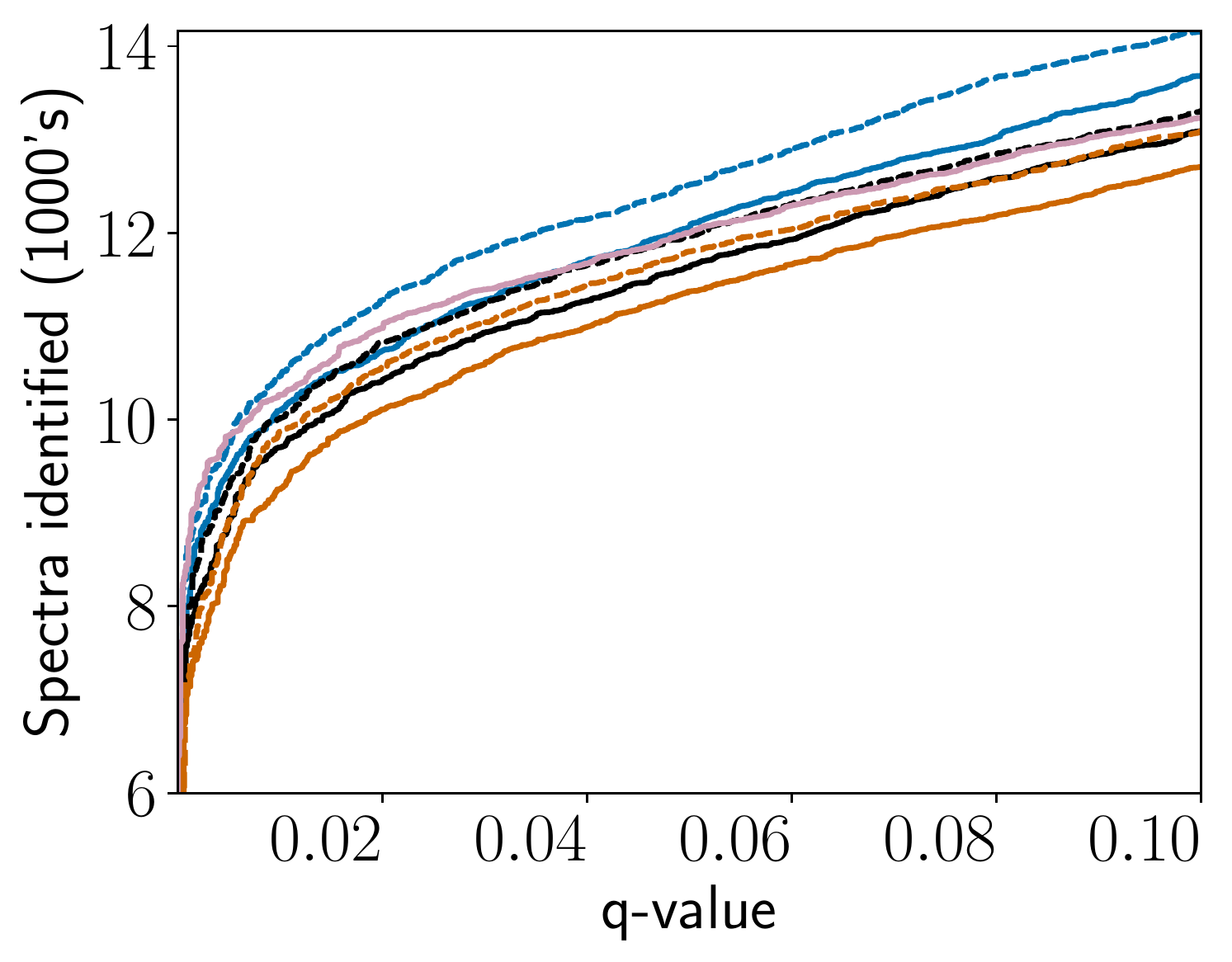}}
  \subfigure[Yeast-2]{\includegraphics[trim=0.45in 0.0in 0.0in 0.05in,
    clip=true,scale=0.3]{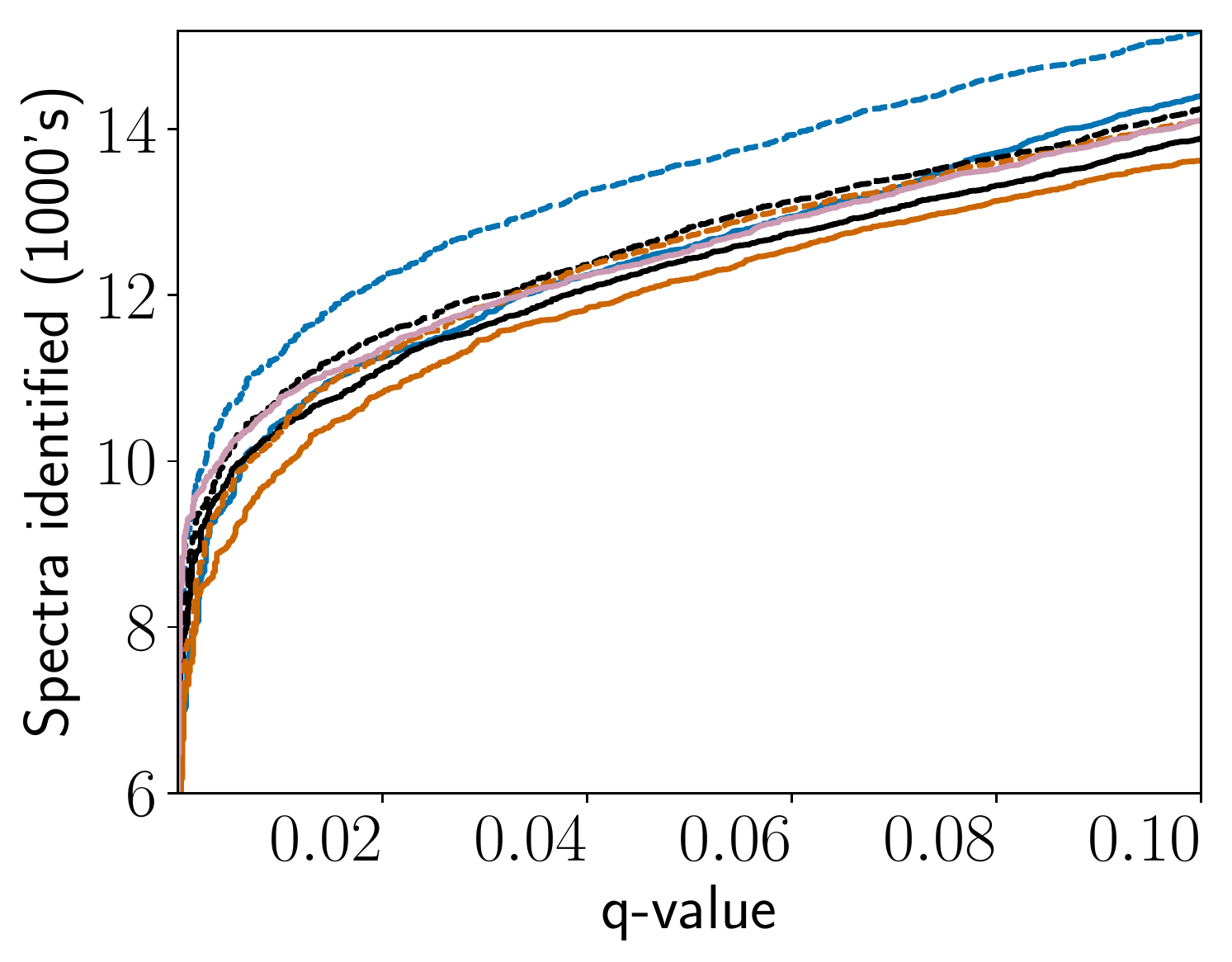}}
  \subfigure[Yeast-3]{\includegraphics[trim=0.45in 0.0in 0.0in 0.05in,
    clip=true,scale=0.3]{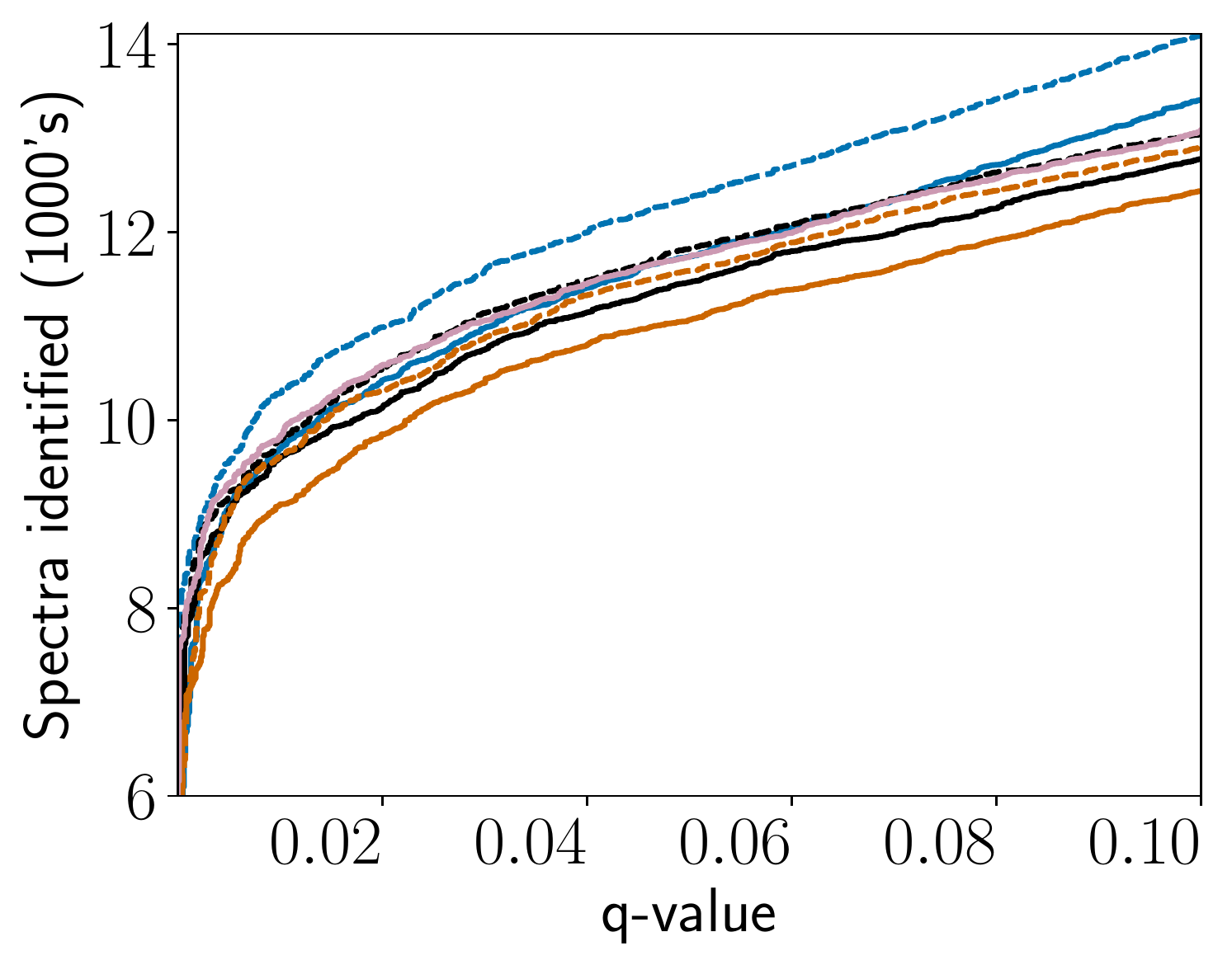}}
  \subfigure[Yeast-4]{\includegraphics[trim=0.0in 0.0in 0.0in 0.05in,
    clip=true,scale=0.3]{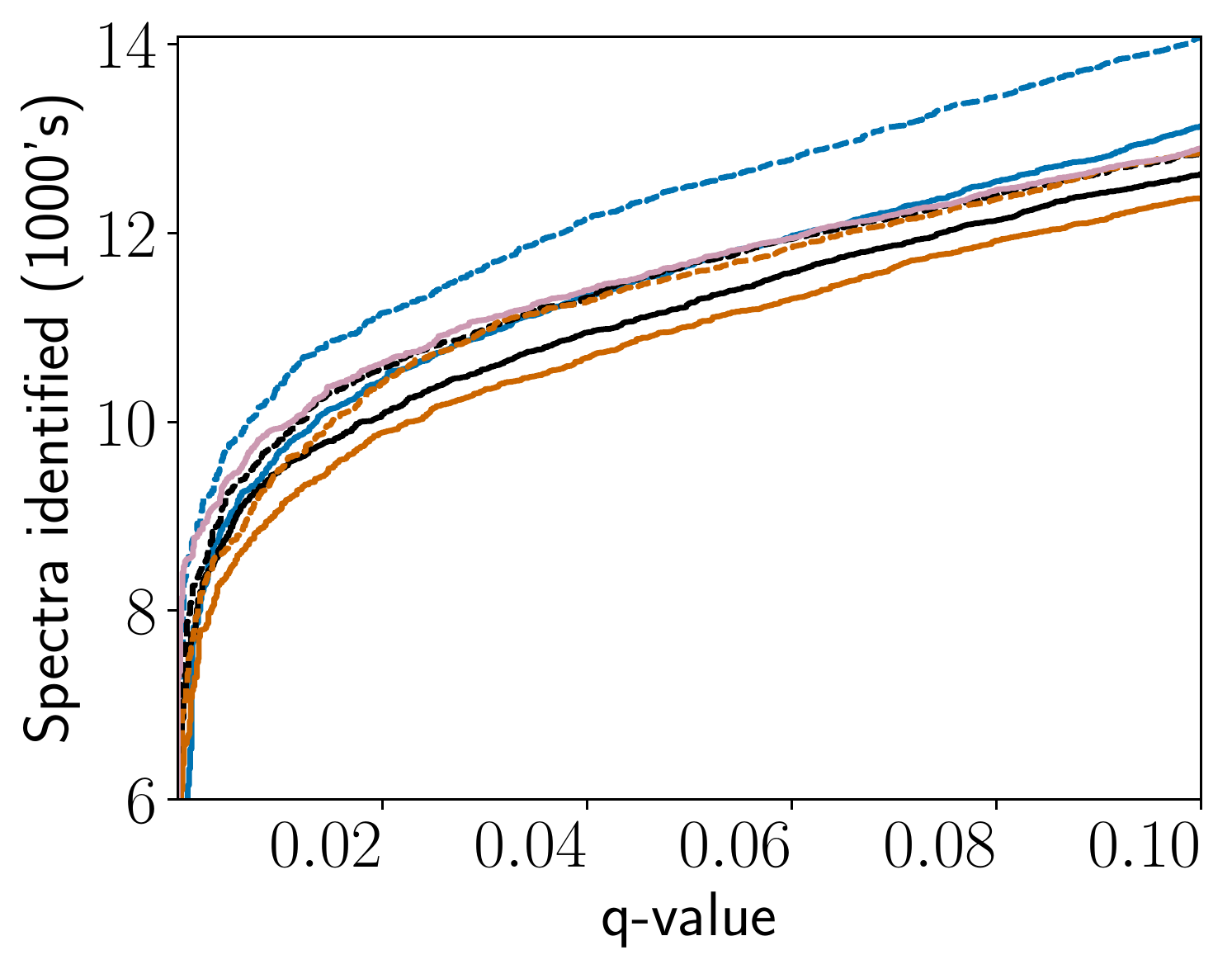}}
  \caption{Search accuracy plots measured by
    $q$-value versus number of spectra identified for worm
    (\emph{C. elegans}) and yeast (\emph{Saccharomyces cerevisiae})
    datasets.  All methods are post-processed using the
    Percolator SVM classifier~\cite{kall:semi-supervised}.  ``DRIP''
    augments the standard set of DRIP features
    with DRIP-Viterbi-path parsed PSM features (described
    in~\cite{halloran2016dynamic}) and ``DRIP Fisher''
    augments the heuristic set with gradient-based DRIP
    features.  ``XCorr,'' ``XCorr $p$-value,'' and ``MS-GF+'' use
    their standard sets of Percolator features (described
    in~\cite{halloran2016dynamic}), while ``XCorr $p$-value Fisher''
    and ``XCorr Fisher'' augment the standard XCorr feature sets with
    gradient-based Theseus features.
  }
  \label{fig:absRanking}
\end{figure}

\begin{figure}[htbp!]
  \centering
  \subfigure{\raisebox{9.0mm}{\includegraphics[trim=2.0in 1.0in 0.3in 1.7in,
    clip=true,scale=0.45]{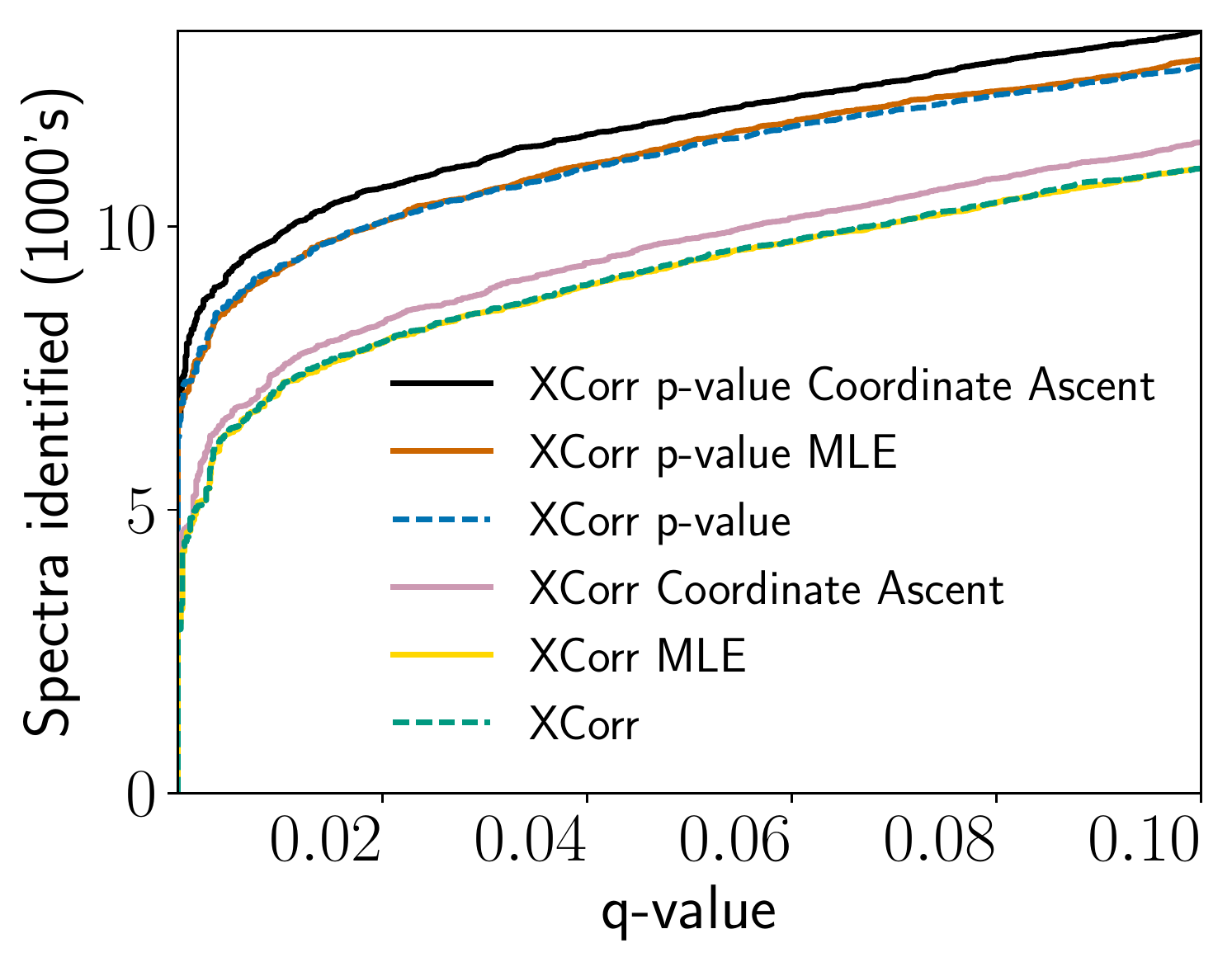}}}
  \subfigure[Yeast-1]{\includegraphics[trim=0.0in 0.0in 0.0in 0.05in,
    clip=true,scale=0.3]{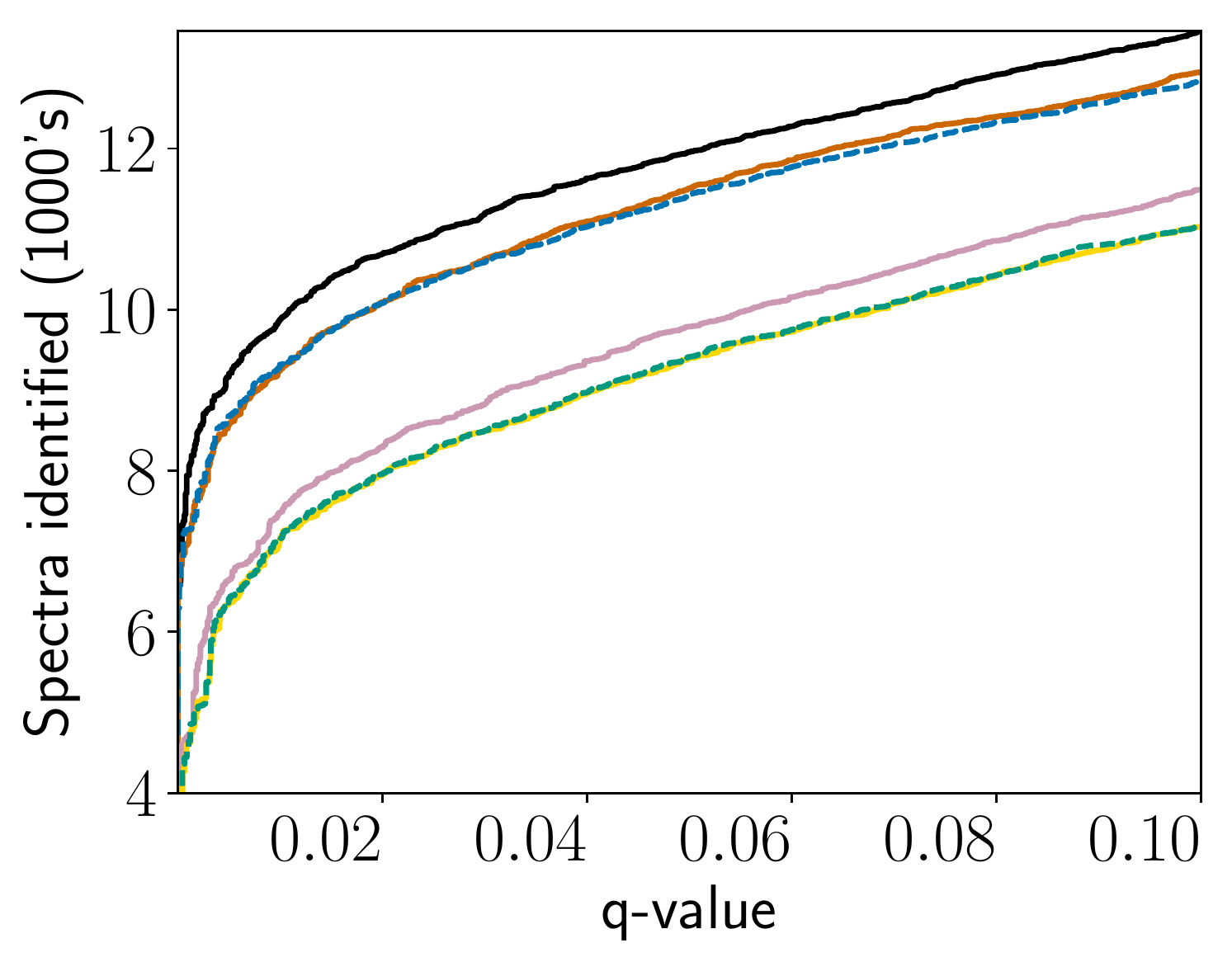}}
  \subfigure[Yeast-2]{\includegraphics[trim=0.45in 0.0in 0.0in 0.05in,
    clip=true,scale=0.3]{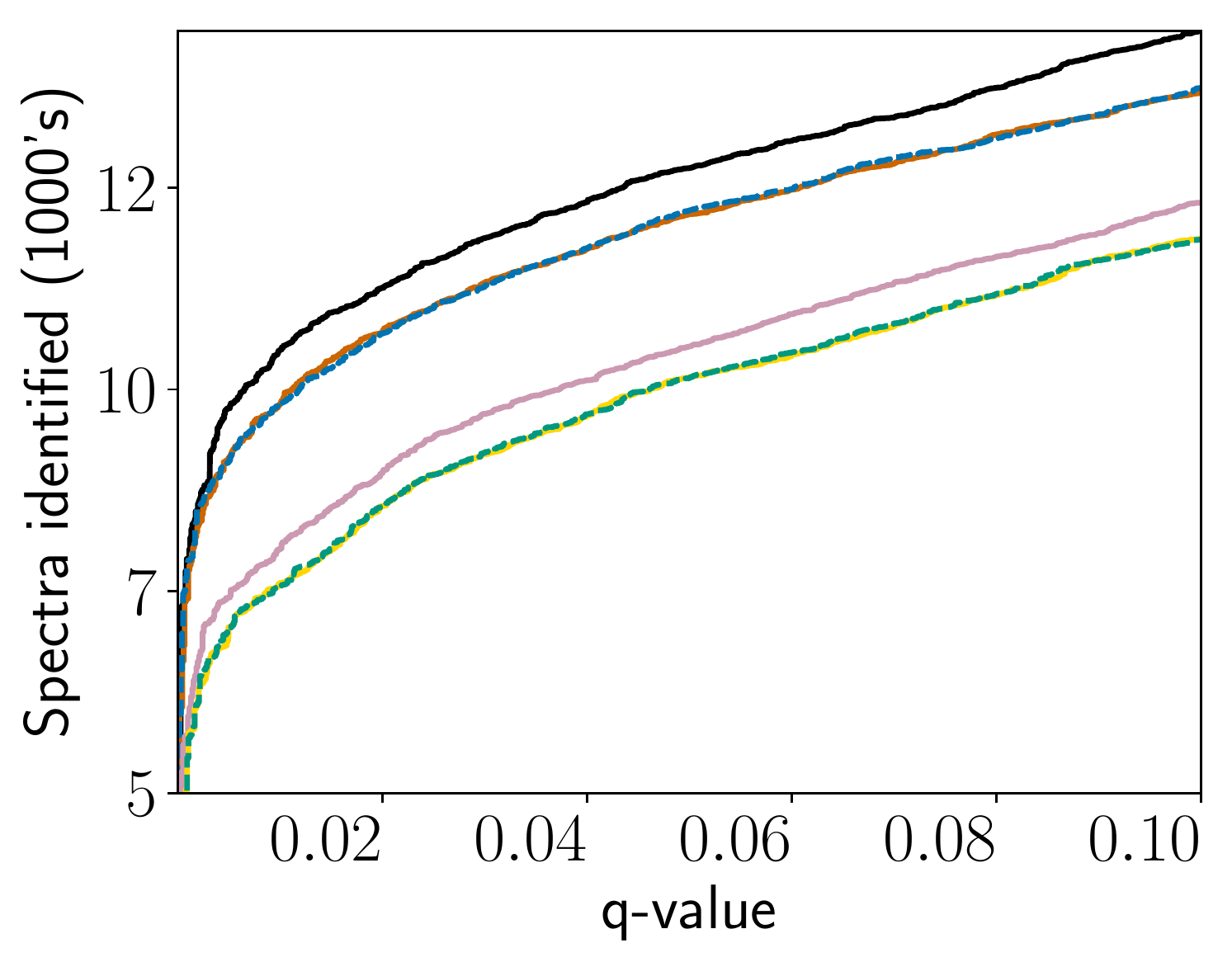}}
  \caption{Search accuracy of Theseus' learned scoring function
    parameters.  Coordinate ascent parameters are learned using
    Algorithm~\ref{algorithm:coordinateAscent} and MLE  parameters are
    learned using gradient ascent.}
  \label{fig:theseusUnsupervisedLearning}
\end{figure}

\subsection{Discussion}
DRIP gradient-based
post-processing improves upon the heuristically derived features in
all cases, and does so substantially on a majority of datasets.  In
the case of the yeast datasets, this distinguishes DRIP
post-processing performance from all competitors and leads to
state-of-the-art identification accuracy.  Furthermore, we note that
both XCorr and XCorr $p$-value post-processing performance are greatly
improved using the gradient-based features derived using Theseus,
raising performance above the highly similar MS-GF+ in several cases.
Particularly noteworthy is the substantial improvement in XCorr
accuracy which, using gradient-based information, is nearly
competitive with its $p$-value counterpart.
Considering the respective runtimes of the underlying search
algorithms, this thus presents a tradeoff for a researcher considering
search time and accuracy.  In
practice, the DRIP and XCorr $p$-value
computations are at least an order of magnitude slower than XCorr
computation in Crux~\cite{wang2016faster}.  Thus, the
presented work not only improves state-of-the-art accuracy, but also
improves the accuracy of simpler, yet significantly faster, search
algorithms.

Owing to max-product inference in graphical models, we also show that
Theseus may be used to effectively learn XCorr model
parameters (Figure~\ref{fig:theseusUnsupervisedLearning}) without
supervision.  Furthermore, we show that XCorr $p$-values are also made
more accurate by training the underlying scoring function for which
$p$-values are computed.  This marks a novel step towards unsupervised
training of uncalibrated scoring functions, as unsupervised learning
has been extensively explored for post-processor recalibration, but
has remained an open problem for MS/MS database-search scoring
functions.
The presented learning framework, as well as the
presented XCorr gradient-based feature representation, may be adapted
by many of the widely scoring functions represented by
Theseus~\cite{craig:tandem, eng:approach,  eng:comet, kim:msgfPlus,
  howbert:computing, wenger2013proteomics,
 mcilwain:crux}.  

Many exciting avenues are open for future work.  Leveraging the large
breadth of graphical models research, we plan to explore other learning paradigms using
Theseus (for instance, estimating other PSMs using $k$-best Viterbi in
order to discriminatively learn parameters using algorithms such as
max-margin learning).  Perhaps most exciting, we plan to further
investigate the peptide-to-observed-spectrum mapping derived from DRIP
Fisher scores.  Under this mapping, we plan to explore learning distance
metrics between PSMs in order to identify proteins from peptides.

\noindent {\bf Acknowledgments}: This work was supported by the
National Center for Advancing Translational Sciences (NCATS), National
Institutes of Health, through grant UL1 TR001860.
\bibliographystyle{plain}
\setcitestyle{numbers, open={[}, close={]}}
\bibliography{fisherKernel_msms_arxiv}

\begin{thebibliography}{10}

\bibitem{benjamini:controlling}
Y.~Benjamini and Y.~Hochberg.
\newblock Controlling the false discovery rate: a practical and powerful
  approach to multiple testing.
\newblock {\em Journal of the Royal Statistical Society B}, 57:289--300, 1995.

\bibitem{craig:tandem}
R.~Craig and R.~C. Beavis.
\newblock Tandem: matching proteins with tandem mass spectra.
\newblock {\em Bioinformatics}, 20:1466--1467, 2004.

\bibitem{dempster:maximum}
A.~P. Dempster, N.~M. Laird, and D.~B. Rubin.
\newblock Maximum likelihood from incomplete data via the {EM} algorithm.
\newblock {\em Journal of the Royal Statistical Society. Series B
  (Methodological)}, 39:1--22, 1977.

\bibitem{elkan2005deriving}
Charles Elkan.
\newblock Deriving tf-idf as a fisher kernel.
\newblock In {\em International Symposium on String Processing and Information
  Retrieval}, pages 295--300. Springer, 2005.

\bibitem{eng:approach}
J.~K. Eng, A.~L. McCormack, and J.~R. {Yates, III}.
\newblock An approach to correlate tandem mass spectral data of peptides with
  amino acid sequences in a protein database.
\newblock {\em Journal of the American Society for Mass Spectrometry},
  5:976--989, 1994.

\bibitem{eng:comet}
Jimmy~K Eng, Tahmina~A Jahan, and Michael~R Hoopmann.
\newblock Comet: An open-source ms/ms sequence database search tool.
\newblock {\em Proteomics}, 13(1):22--24, 2013.

\bibitem{halloran2014uai-drip}
John~T. Halloran, Jeff~A. Bilmes, and William~S. Noble.
\newblock Learning peptide-spectrum alignment models for tandem mass
  spectrometry.
\newblock In {\em Uncertainty in Artificial Intelligence (UAI)}, Quebec City,
  Quebec Canada, July 2014. AUAI.

\bibitem{halloran2016dynamic}
John~T Halloran, Jeff~A Bilmes, and William~S Noble.
\newblock Dynamic bayesian network for accurate detection of peptides from
  tandem mass spectra.
\newblock {\em Journal of Proteome Research}, 15(8):2749--2759, 2016.

\bibitem{howbert:computing}
J~Jeffry Howbert and William~S Noble.
\newblock Computing exact p-values for a cross-correlation shotgun proteomics
  score function.
\newblock {\em Molecular \& Cellular Proteomics}, pages mcp--O113, 2014.

\bibitem{jaakkolaFisherKernelNips1998}
T.~Jaakkola and D.~Haussler.
\newblock Exploiting generative models in discriminative classifiers.
\newblock In {\em Advances in Neural Information Processing Systems},
  Cambridge, MA, 1998. MIT Press.

\bibitem{jaakkola1999using}
Tommi~S Jaakkola, Mark Diekhans, and David Haussler.
\newblock Using the fisher kernel method to detect remote protein homologies.
\newblock In {\em ISMB}, volume~99, pages 149--158, 1999.

\bibitem{kall:semi-supervised}
L.~K\"{a}ll, J.~Canterbury, J.~Weston, W.~S. Noble, and M.~J. MacCoss.
\newblock A semi-supervised machine learning technique for peptide
  identification from shotgun proteomics datasets.
\newblock {\em Nature Methods}, 4:923--25, 2007.

\bibitem{keich2015improved}
Uri Keich, Attila Kertesz-Farkas, and William~Stafford Noble.
\newblock Improved false discovery rate estimation procedure for shotgun
  proteomics.
\newblock {\em Journal of proteome research}, 14(8):3148--3161, 2015.

\bibitem{keich2014importance}
Uri Keich and William~Stafford Noble.
\newblock On the importance of well-calibrated scores for identifying shotgun
  proteomics spectra.
\newblock {\em Journal of proteome research}, 14(2):1147--1160, 2014.

\bibitem{kim:msgfPlus}
Sangtae Kim and Pavel~A Pevzner.
\newblock Ms-gf+ makes progress towards a universal database search tool for
  proteomics.
\newblock {\em Nature communications}, 5, 2014.

\bibitem{mcilwain:crux}
Sean McIlwain, Kaipo Tamura, Attila Kertesz-Farkas, Charles~E Grant, Benjamin
  Diament, Barbara Frewen, J~Jeffry Howbert, Michael~R Hoopmann, Lukas
  K{\"a}ll, Jimmy~K Eng, et~al.
\newblock Crux: rapid open source protein tandem mass spectrometry analysis.
\newblock {\em Journal of proteome research}, 2014.

\bibitem{pearl:probabilistic}
J.~Pearl.
\newblock {\em Probabilistic Reasoning in Intelligent Systems : Networks of
  Plausible Inference}.
\newblock Morgan Kaufmann, 1988.

\bibitem{spivak:improvements}
M.~Spivak, J.~Weston, L.~Bottou, L.~K\"all, and W.~S. Noble.
\newblock Improvements to the {Percolator} algorithm for peptide identification
  from shotgun proteomics data sets.
\newblock {\em Journal of Proteome Research}, 8(7):3737--3745, 2009.
\newblock PMC2710313.

\bibitem{spivak:direct}
M.~Spivak, J.~Weston, D.~Tomazela, M.~J. MacCoss, and W.~S. Noble.
\newblock Direct maximization of protein identifications from tandem mass
  spectra.
\newblock {\em Molecular and Cellular Proteomics}, 11(2):M111.012161, 2012.
\newblock PMC3277760.

\bibitem{wang2016faster}
Shengjie Wang, John~T Halloran, Jeff~A Bilmes, and William~S Noble.
\newblock Faster and more accurate graphical model identification of tandem
  mass spectra using trellises.
\newblock {\em Bioinformatics}, 32(12):i322--i331, 2016.

\bibitem{wenger2013proteomics}
C.~D. Wenger and J.~J. Coon.
\newblock A proteomics search algorithm specifically designed for
  high-resolution tandem mass spectra.
\newblock {\em Journal of proteome research}, 2013.

\end{thebibliography}
\appendix
\section{DRIP Fisher Score Derivation}\label{appendix:dripFisherKernelDerivation}
Following the discussion in Section~\ref{section:dripFisherScores}, $\frac{\partial}{\partial
  \mumz(k)} \log p(s |
x,\theta) = \frac{1}{p(s | x,\theta)} \frac{\partial}{\partial
  \mumz(k)} p(s |x, \theta)$ and we have $\frac{\partial}{\partial
  \mumz(k)} p(s |x, \theta)$
{\small
\begin{align}
=& 
\frac{\partial}{\partial \mumz(k)}
\sum_{i_{1:\dripT}, \delta_{1:\dripT}} p(i_{1:\dripT}, \delta_{1:\dripT} | \theta)
=
\sum_{i_{1:\dripT}, \delta_{1:\dripT} : K_t = k, 1 \leq
  t \leq \dripT} \frac{\partial}{\partial \mumz(k)}p(i_{1:\dripT}, \delta_{1:\dripT}
| \theta) \nonumber \\
=&
\sum_{i_{1:\dripT}, \delta_{1:\dripT}} \indicator_{\{ K_t = k \}}
\prod_{t: K_t \neq k}
\phi (\delta_t, K_{t-1}, i_t, i_{t-1})
\frac{\partial}{\partial \mumz(k)} 
\prod_{t: K_t = k} \phi (\delta_t, K_{t-1}, i_t,
i_{t-1}) \nonumber \\
=&
\sum_{i_{1:\dripT}, \delta_{1:\dripT}} \indicator_{\{ K_t = k \}}
\prod_{t: K_t \neq k}
\phi (\delta_t, K_{t-1}, i_t, i_{t-1})
\left (\prod_{t: K_t = k}\frac{\phi (\delta_t, K_{t-1},
  i_t, i_{t-1})}{p(\Obsmz_t | K_t)} \right)
\left (\frac{\partial}{\partial \mumz(k)}  
\prod_{t: K_t = k}p(\Obsmz_t | K_t) \right) \nonumber\\ 
=&
\sum_{i_{1:T}, \delta_{1:T}} \indicator_{\{ K_t = k \}}
\prod_{t}
\phi (\delta_t, K_{t-1}, i_t, i_{t-1})
\left ( \prod_{t: K_t = k}\frac{1}{p(\Obsmz_t
    | K_t)} \right )
\left (\frac{\partial}{\partial \mumz(k)}  
\prod_{t: K_t = k} p(\Obsmz_t | K_t) \right ) \nonumber \\
=&
\sum_{i_{1:\dripT}, \delta_{1:\dripT}} \indicator_{\{ K_t = k \}}
p(s| x,\theta)
\left ( \prod_{t: K_t = k}\frac{1}{p(\Obsmz_t | K_t)} \right )
\left (\frac{\partial}{\partial \mumz(k)}  
\prod_{t: K_t = k}p(\Obsmz_t | K_t)
\right ), \nonumber
\end{align}
}
where
{\small
\begin{align*}
\frac{\partial}{\partial \mumz(k)}  
\prod_{t: K_t = k}p(\Obsmz_t | K_t) =&
\left (
\prod_{t: K_t = k}p(\Obsmz_t | K_t) 
\right )
\left (
\sum_{t: K_t = k}
\frac{
\frac{\partial}{\partial \mumz(k)}
  \sum_{i_t = 0}^{1}p(i_t)p(\Obsmz_t | K_t, i_t)}
{p(\Obsmz_t | K_t)}
\right)
\\
=& \left (
\prod_{t: K_t = k} p(\Obsmz_t | K_t)
\right )
\left (
\sum_{t: K_t = k} 
  \frac{p(i_t = 0)p(\Obsmz_t | K_t, i_t = 0)\frac{(\Obsmz_t -
    \mumz(k))}{\mzVar}}
{p(\Obsmz_t | K_t)}
 \right)
\\
=& \left (
\prod_{t: K_t = k}
p(\Obsmz_t | K_t)
\right )
\left (
\sum_{t: K_t = k} 
p(i_t = 0 | K_t, \Obsmz_t)\frac{(\Obsmz_t -
    \mumz(k))}{\mzVar}
 \right )
.
\end{align*}
}
{\small
\begin{align}
\Rightarrow 
\frac{\partial}{\partial \mumz(k)} \log p(s | x, \theta) 
=&
\frac{1}{p(s|x,\theta)}\sum_{i_{1:\dripT}, \delta_{1:\dripT}} 
\indicator_{\{ K_t =
  k \}}
p(i_{1:\dripT,} K_{1:\dripT} | \theta)
\sum_{t: K_t = k} 
  p(i_t = 0 | K_t, \Obsmz_t)\frac{(\Obsmz_t -
    \mumz(k))}{\mzVar} \nonumber\\
=&\frac{1}{p(s|x,\theta)} \sum_{t=1}^\dripT
p(i_{t}, K_{t} = k | \theta)
  p(i_t = 0 | K_t, \Obsmz_t)\frac{(\Obsmz_t -
    \mumz(k))}{\mzVar} \nonumber\\
=& \sum_{t = 1}^\dripT
p(i_{t}, K_{t} = k | s, \theta)
  p(i_t = 0 | K_t, \Obsmz_t)\frac{(\Obsmz_t -
    \mumz(k))}{\mzVar},\label{eqn:fisherScoreMzMean}
\end{align}}
\noindent where we equivalently write $p(s | x,\theta) = p(i_{1:\dripT},\delta_{1:\dripT} |
\theta) = p(i_{1:\dripT}, K_{1:\dripT} | \theta)$ due to the deterministic
relationship $\delta_t = K_t - K_{t-1}$.  

\section{Theseus Unsupervised Learning using Coordinate Ascent}\label{appendix:monotonicProof}
Using the model's Fisher scores, Theseus parameters $\theta$ may be
learned via maximum likelihood estimation.  We present an alternate
learning algorithm (compared to maximum likelihood learning in Section~\ref{section:results}).  Let $s^1, s^2, \dots, s^n$
be a dataset of spectra and define $J(x^1, \dots, x^n, \theta) =
\sum_{i=1}^n \log p(s^i, x^i | \theta)$.  Optimizing this objective
function, Theseus' coordinate ascent learning algorithm is defined in
Algorithm~\ref{algorithm:coordinateAscent} where, rather than relying
on training labels, we use max-product inference to infer the most
probable PSM for each spectrum given the current iteration's
parameters, then maximize the log-likelihood with respect to $\theta$
given the most likely PSMs.  We now prove that
Algorithm~\ref{algorithm:coordinateAscent} converges monotonically.
\begin{theorem}\label{theorem:coordinateAscentMonotoneConvergence}
Algorithm~\ref{algorithm:coordinateAscent} converges
monotonically to a local optimum.
\end{theorem}
\begin{proof}
We need to show that the objective function $J$ is nondecreasing with
each iteration of the algorithm.  Denote the learned parameters at iteration
$k$ of the algorithm as $\theta_k$ and define $\hat{x}^i_k =
\argmax_{x^i \in \cP}\log p(s^i, x^i | \theta_{k-1}), \theta_k =
\argmax_{\theta} J(\hat{x}^1_k, \dots, \hat{x}^1_k, \theta)$.  We
thus have
\begin{align*}
J(\hat{x}^1_{k},\dots, \hat{x}^n_{k}, \theta_{k}) &\geq J(\hat{x}^1_{k},\dots, \hat{x}^n_{k},
\theta_{k-1})\\
J(\hat{x}^1_{k},\dots, \hat{x}^n_{k}, \theta_{k-1}) &\geq J(x^1,\dots, x^n,
\theta_{k-1}), \; \forall x^1,\dots,x^n \in \cP\\
\Rightarrow J(\hat{x}^1_{k},\dots, \hat{x}^n_{k}, \theta_{k}) &\geq
J(\hat{x}^1_{k-1},\dots, \hat{x}^n_{k-1}, \theta_{k-1})
\end{align*}
\end{proof}

\vspace{-0.2in}
\section{Impact of Recalibration over Standard DRIP Search}\label{appendix:impact}
\vspace{-0.2in}

\begin{table}[htbp!]
\caption{{\small Percent improvement over uncalibrated search results for the
  DRIP methods plotted in Figure~\ref{fig:dripAbsRanking}, at an
  FDR threshold $t=1\%$.  Largest improvement highlighted in bold.
  Note that at this FDR threshold, Percolator post-processing using a
  standard set of features may result in diminished performance
  (Worm-3).}}
  \label{table:datasets}
\centering
\begin{tabular}{cccc}
\hline
Data set&DRIP&DRIP Heuristic&DRIP Fisher\\
\hline
Yeast-1& 5.4 & 10.7 & {\bf 14.8}\\
Yeast-2& 5.2 & 8.3 & {\bf 16.6}\\
Yeast-3& 9.2 & 10.9 & {\bf 17.7}\\
Yeast-4& 3.4 &  7.5 & {\bf 15.1}\\
Worm-1& 10.1 & 17.4 & {\bf 20.8} \\
Worm-2& 1.1 &  6.7 & {\bf 11.3 }\\
Worm-3& -5.1 & 7.2 & {\bf 11 }\\
Worm-4& 0.4 & 9.9 & {\bf 16}\\\hline
Average& 3.7 & 9.8 & {\bf 15.4}\\
\hline
\end{tabular} 
\end{table}

\begin{figure}[htbp!]
\vspace{-0.2in}
  \centering
  \subfigure{\raisebox{8.5mm}{\includegraphics[trim=2.7in 1.0in 0.2in 1.85in,
    clip=true,scale=0.5]{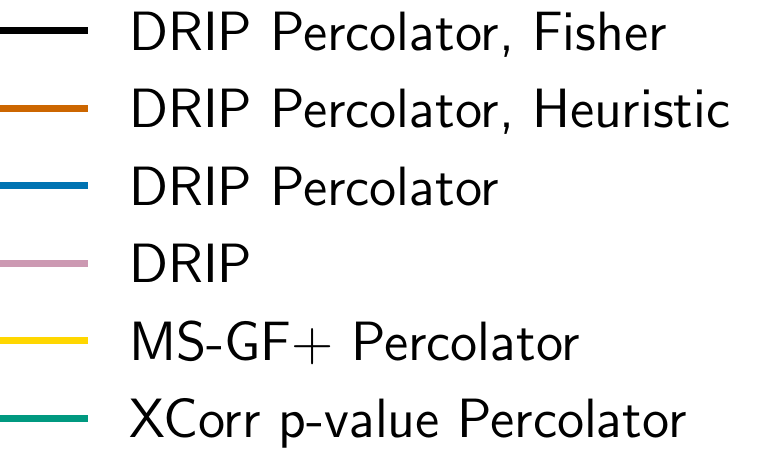}}}
  \subfigure[Worm-1]{\includegraphics[trim=0.0in 0.0in 0.0in 0.05in,
    clip=true,scale=0.28]{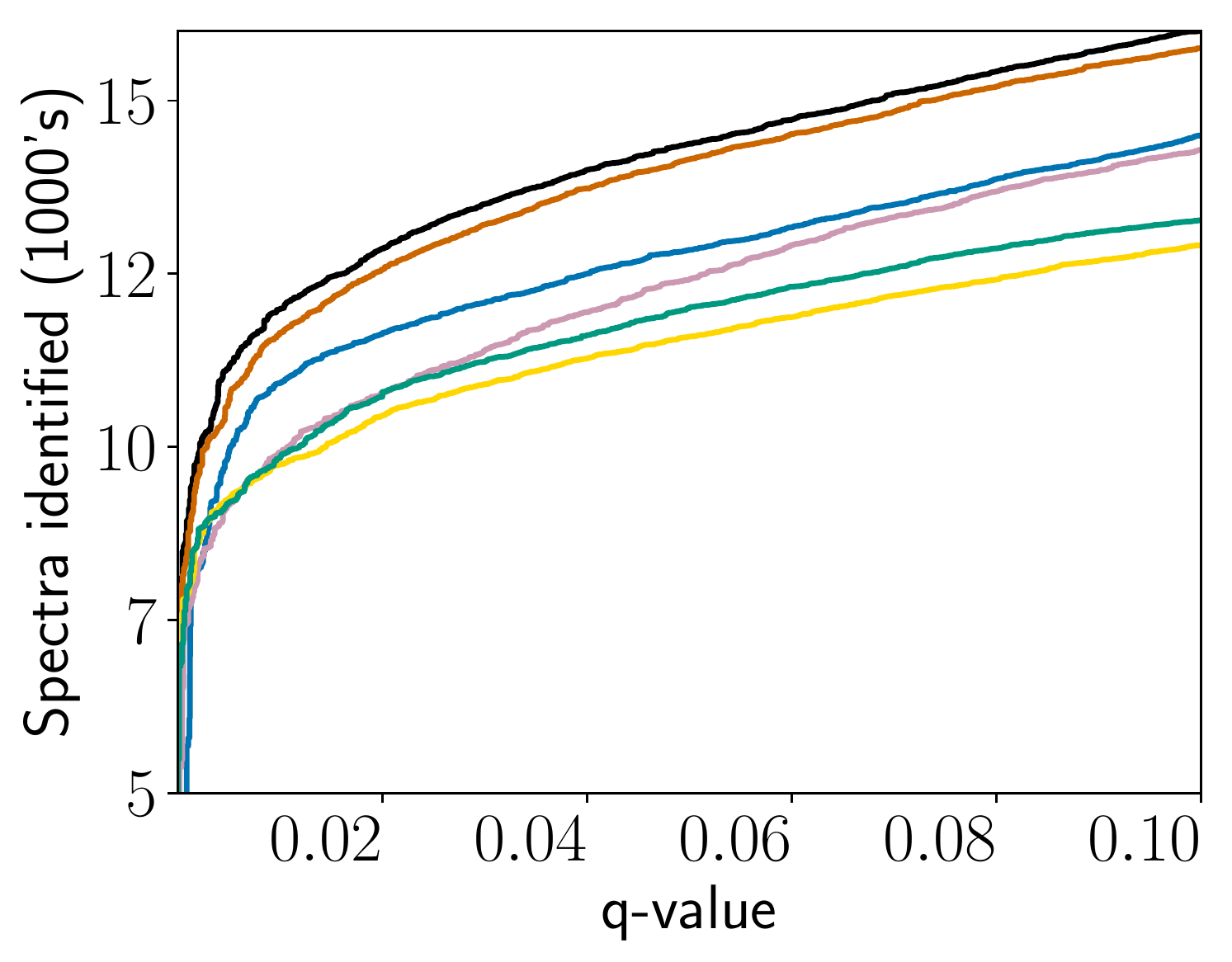}}
  \subfigure[Worm-2]{\includegraphics[trim=0.45in 0.0in 0.0in 0.05in,
    clip=true,scale=0.28]{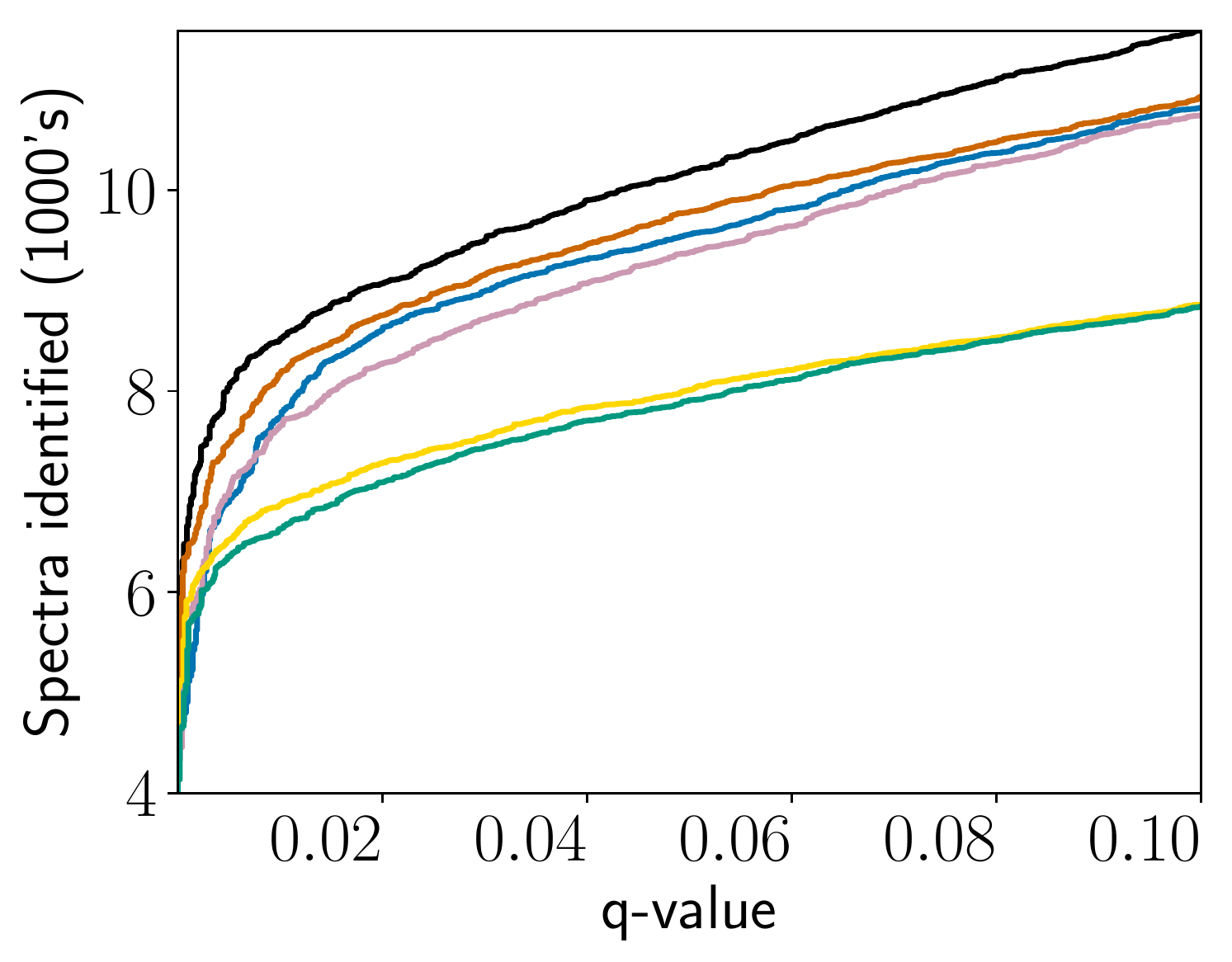}}
  \subfigure[Worm-3]{\includegraphics[trim=0.0in 0.0in 0.0in 0.05in,
    clip=true,scale=0.28]{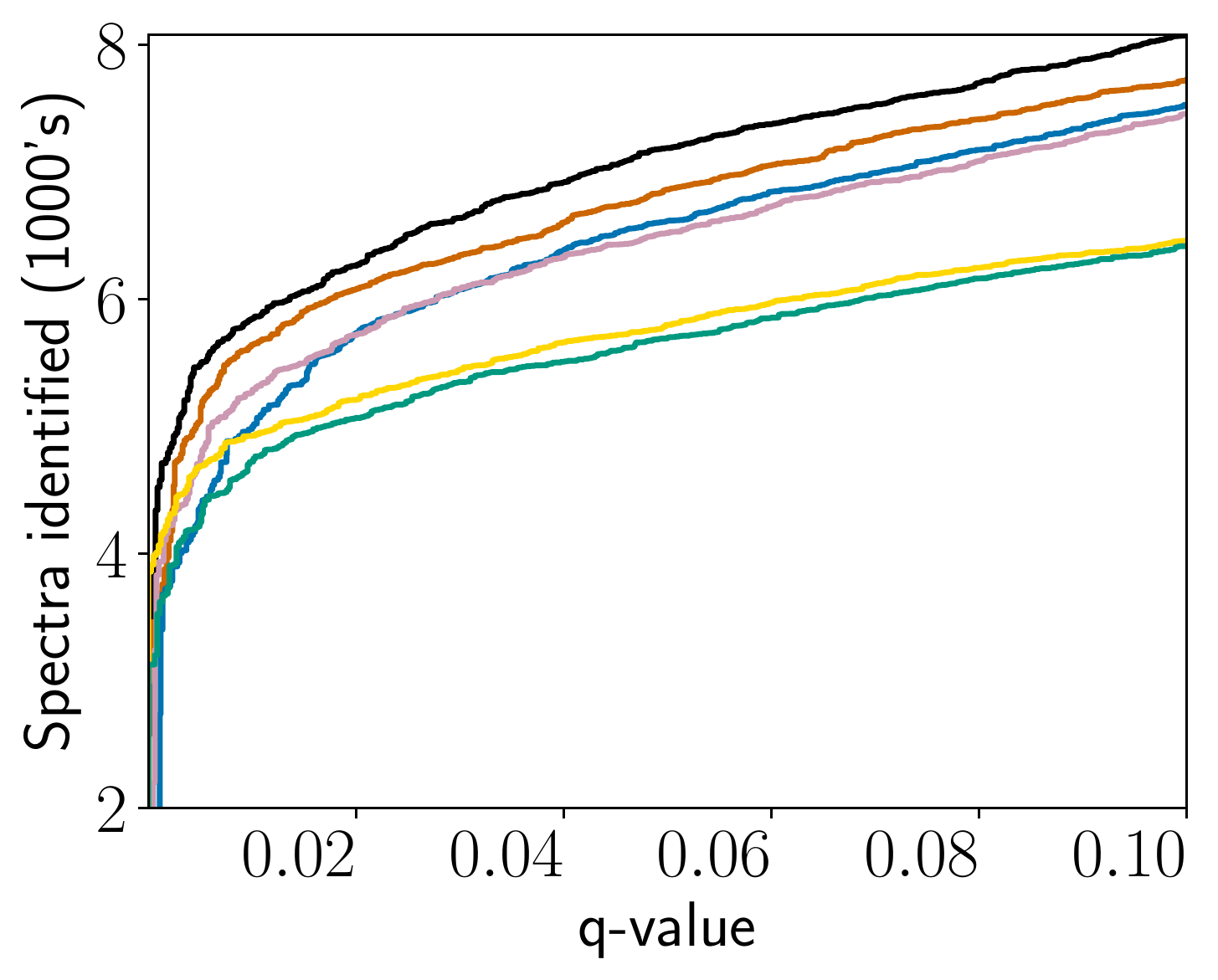}}
  \subfigure[Worm-4]{\includegraphics[trim=0.45in 0.0in 0.0in 0.05in,
    clip=true,scale=0.28]{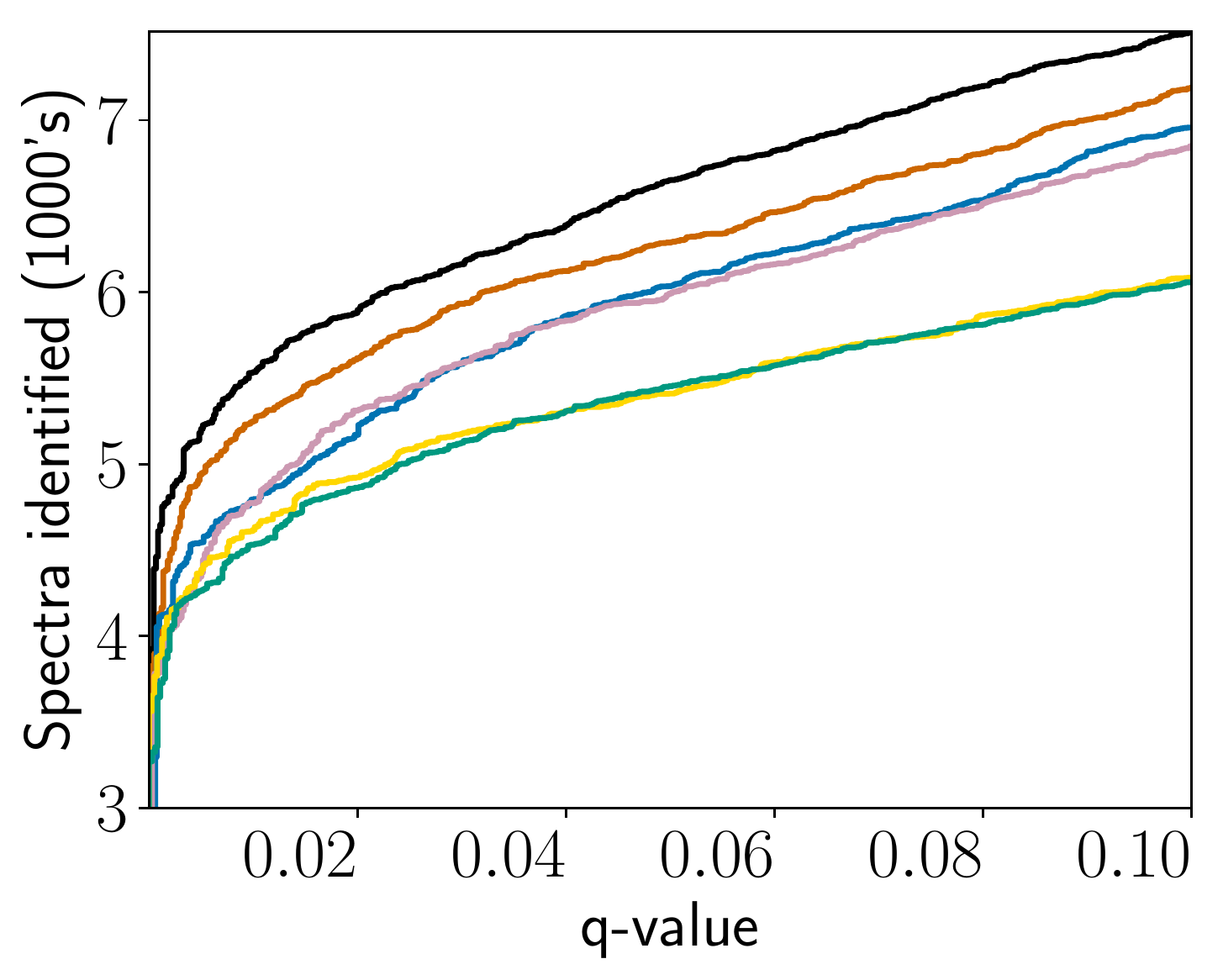}}
  \subfigure[Yeast-1]{\includegraphics[trim=0.45in 0.0in 0.0in 0.05in,
    clip=true,scale=0.28]{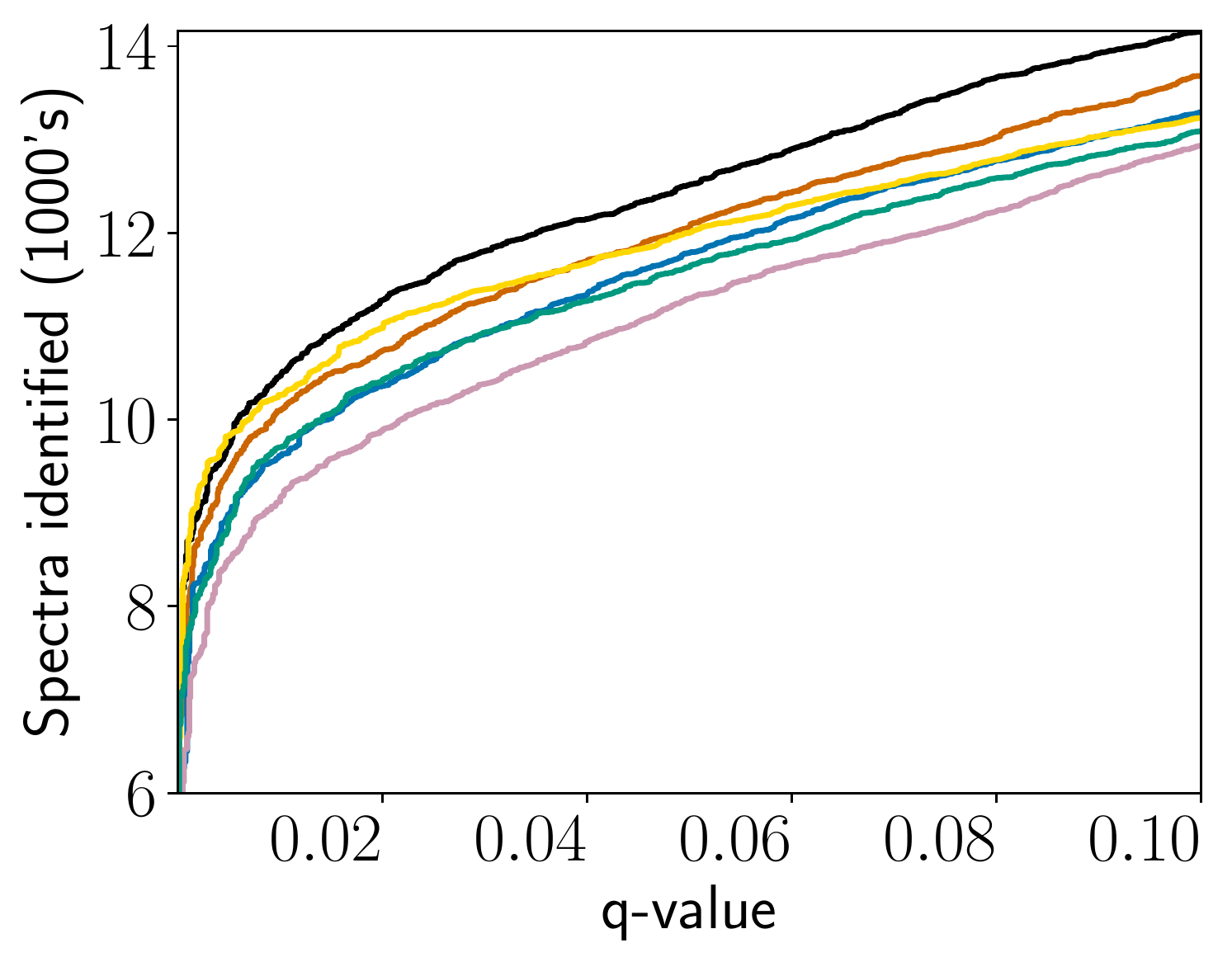}}
  \subfigure[Yeast-2]{\includegraphics[trim=0.0in 0.0in 0.0in 0.05in,
    clip=true,scale=0.28]{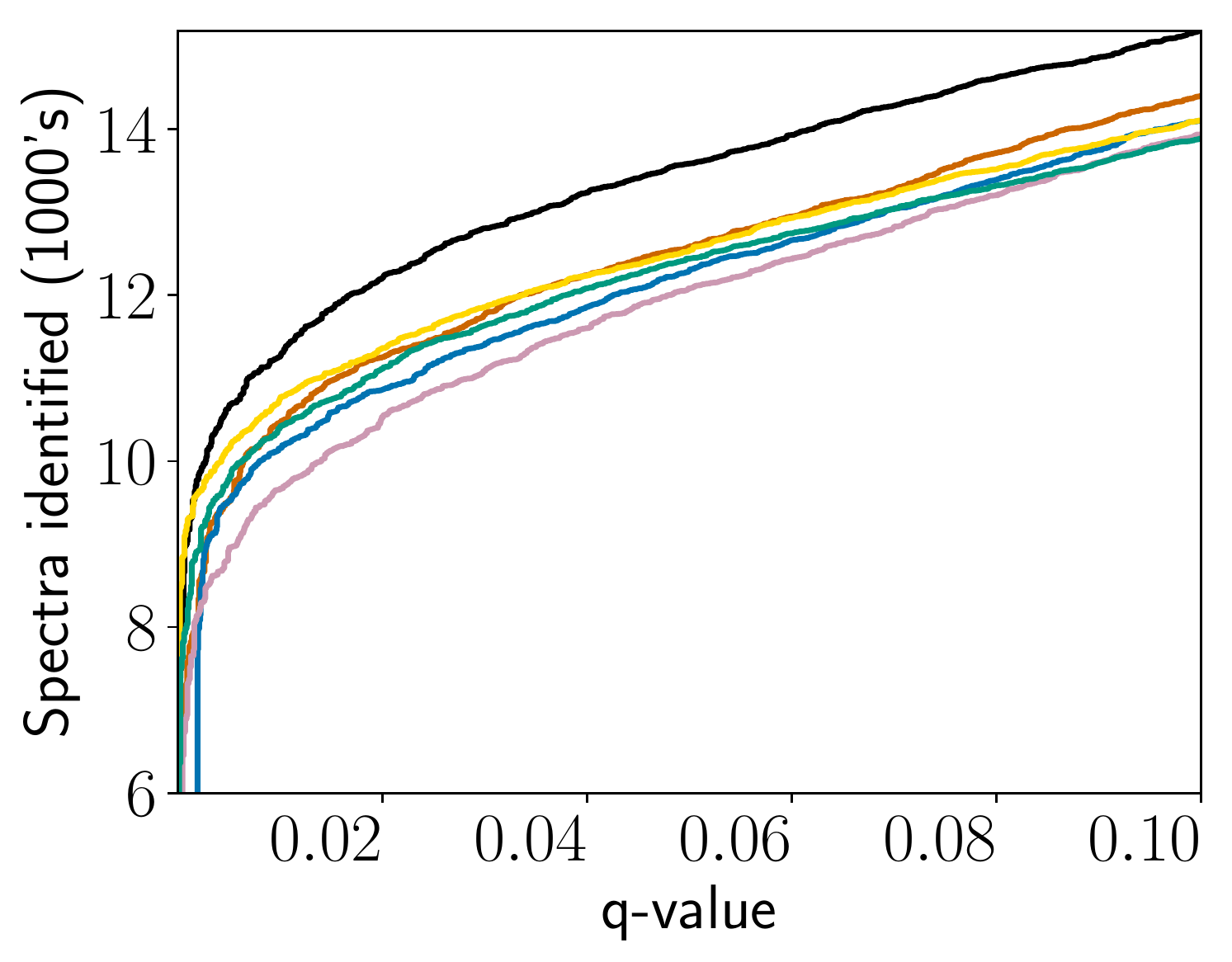}}
  \subfigure[Yeast-3]{\includegraphics[trim=0.45in 0.0in 0.0in 0.05in,
    clip=true,scale=0.28]{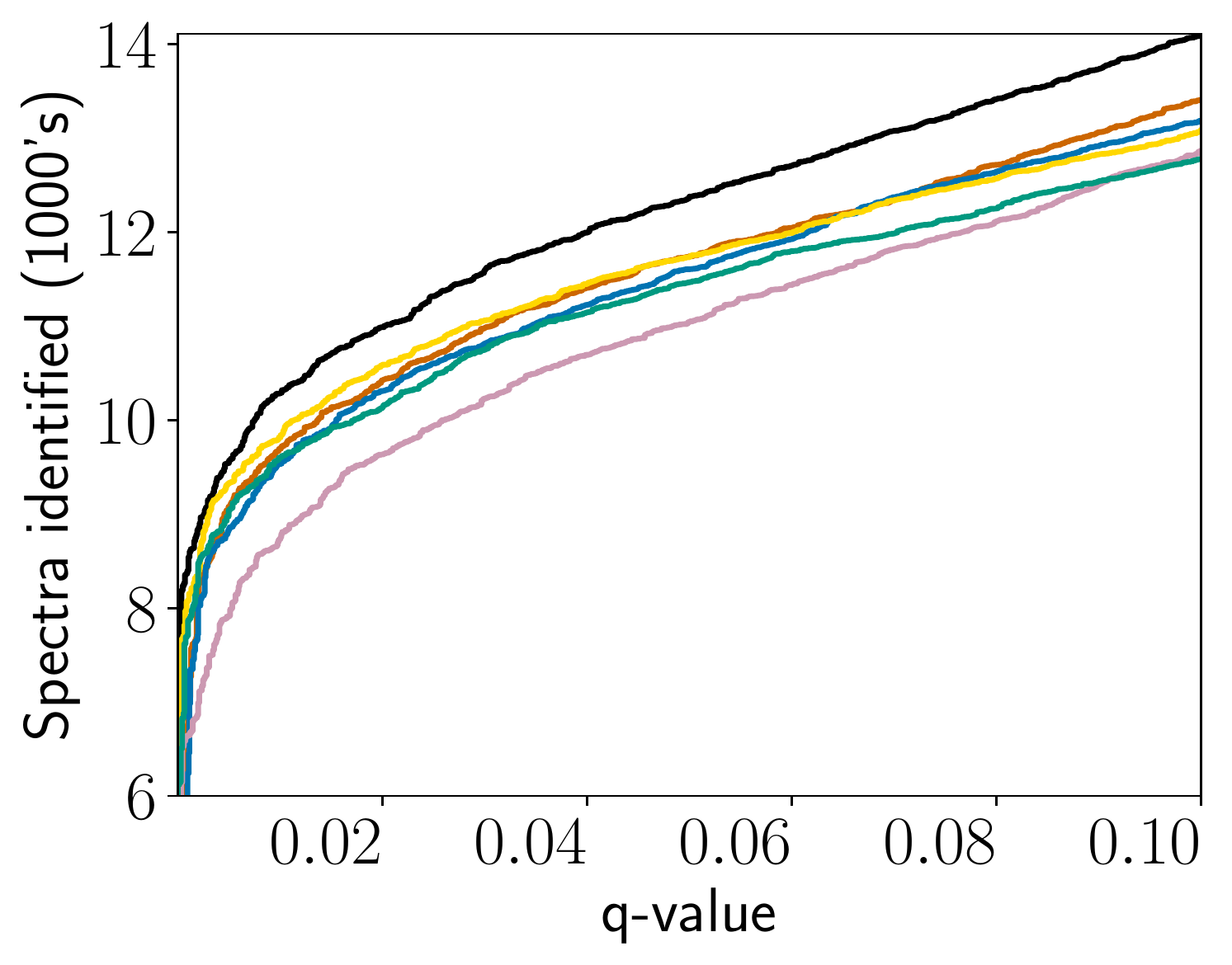}}
  \subfigure[Yeast-4]{\includegraphics[trim=0.45in 0.0in 0.0in 0.05in,
    clip=true,scale=0.28]{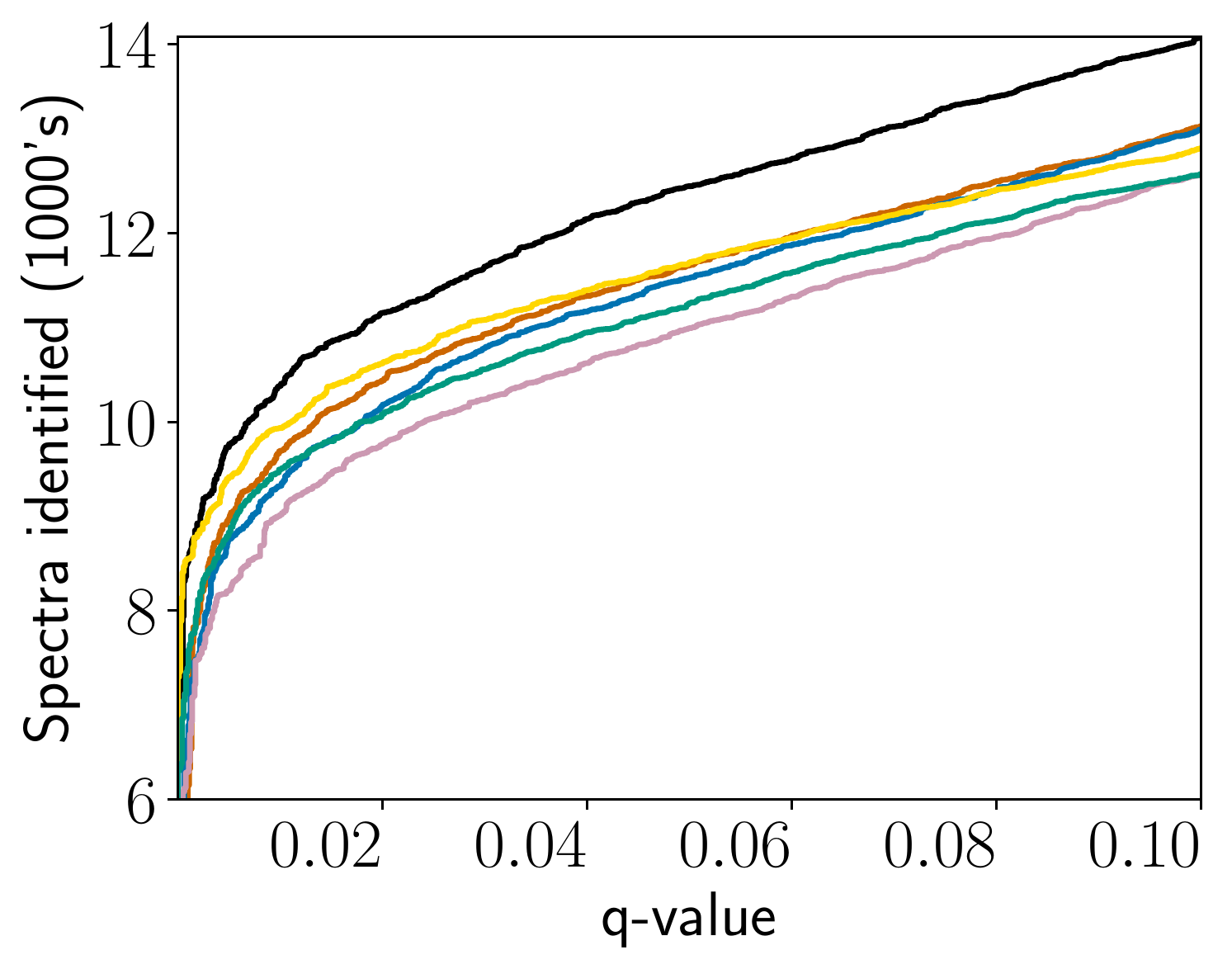}}
  \caption{{\small Performance increase of DRIP search after recalibration.
    Methods denoted by ``Percolator'' are post-processed using the
    Percolator SVM classifier~\cite{kall:semi-supervised}, otherwise
    the raw PSM scores of the denoted search algorithm are used for
    identification.  ``DRIP
    Percolator'' uses the standard set of DRIP PSM features described
    in~\cite{halloran2016dynamic}, ``DRIP Percolator, Heuristic''
    augments the standard set with DRIP-Viterbi-path parsed PSM features described
    in~\cite{halloran2016dynamic}, and ``DRIP Percolator, Fisher''
    augments the Heuristic set with the gradient-based DRIP features
    to the standard.  XCorr $p$-value and MS-GF+ use their standard
    set of Percolator features, described
    in~\cite{halloran2016dynamic}.  Search accuracy plots measured by
    $q$-value versus number of spectra identified for yeast
    (\emph{Saccharomyces cerevisiae}) and worm (\emph{C. elegans})
    datasets.}}
  \label{fig:dripAbsRanking}
\end{figure}
\end{document}